\documentclass[12pt, letterpaper]{article}
\usepackage{amsmath,amsthm,amssymb}
\usepackage{listings}
\usepackage{graphicx}
\usepackage{python}
\usepackage{caption}
\usepackage{subcaption}
 \usepackage{here}
\usepackage{color}

\usepackage[letterpaper, margin=1in]{geometry}
\usepackage{setspace}
\usepackage{titling,lipsum}
\usepackage[dvipsnames,svgnames,table]{xcolor}
\usepackage{hyperref}
\usepackage{multicol}
\usepackage{multirow}
\usepackage{mathtools}
\usepackage{float}
\usepackage{booktabs}
\usepackage{natbib}
\usepackage{mathtools}
\usepackage{preview}
\usepackage{standalone}
\usepackage[capitalize,nameinlink]{cleveref}
\usepackage[toc,page]{appendix}
\usepackage{algorithm}
\usepackage[noend]{algpseudocode}
\usepackage{siunitx,array,multirow}
\usepackage{tikz}
{\theoremstyle{definition}}

\hypersetup{colorlinks,linkcolor={Blue},citecolor={Blue},urlcolor={Blue}}
\theoremstyle{definition}

\newtheorem{prop}{Proposition}

\newcommand{\bu}{\bold{u}}
\newcommand{\by}{\bold{y}}
\newcommand{\bx}{\bold{x}}
\newcommand{\bY}{\bold{Y}}
\newcommand{\bX}{\bold{X}}
\newcommand{\bH}{\bold{H}}
\newcommand{\px}{\pmb{x}}
\newcommand{\py}{\pmb{y}}
\newcommand{\pz}{\pmb{z}}
\newcommand{\pu}{\pmb{u}}
\newcommand{\bU}{\bold{U}}
\newcommand{\pbm}{\pmb{m}}
\newcommand{\bq}{\bold{q}}
\newcommand{\bp}{\bold{p}}
\newcommand{\bz}{\bold{z}}
\newcommand{\bm}{\bold{m}}

\newcommand{\argmin}{\operatornamewithlimits{argmin}}

\newcommand{\bmu}{\boldsymbol{\mu}}
\newcommand{\bgamma}{\boldsymbol{\gamma}}
\newcommand{\bbeta}{\boldsymbol{\beta}}
\newcommand{\bveps}{\boldsymbol{\varepsilon}}

\algnewcommand\Input{\item[\textbf{Input:}]}
\algnewcommand\Output{\item[\textbf{Output:}]}
\algnewcommand\Initialize{\item[\textbf{Initialize :}]}

\title{Robust Extrinsic Regression Analysis for Manifold Valued Data}

\author{
  Hwiyoung Lee\\
  \texttt{hwiyoung.lee@stat.fsu.edu}
}

\date{Department of Statistics, Florida State University\\ January 28, 2021}

\allowdisplaybreaks
\begin{document}

\maketitle

\begin{abstract}
\onehalfspacing
\selectfont
Recently, there has been a growing need in analyzing data on manifolds owing to their important role in diverse fields of science and engineering. In the literature of manifold-valued data analysis up till now, however, only a few works have been carried out concerning the robustness of estimation against noises, outliers, and other sources of perturbations. In this regard, we introduce a novel extrinsic framework for analyzing manifold valued data in a robust manner. First, by extending the notion of the geometric median, we propose a new robust location parameter on manifolds, so-called the extrinsic median. A robust extrinsic regression method is also developed by incorporating the conditional extrinsic median into the classical local polynomial regression method. We present the Weiszfeld's algorithm for implementing the proposed methods. The promising performance of our approach against existing methods is illustrated through simulation studies. \vspace{.5em}

 \end{abstract}

\section{Introduction}
Over the past few decades, analyzing data taking vales in non-Euclidean spaces, mostly nonlinear manifolds has attracted increased attention in a wide range of applications, because it allows a richer and more accurate statistical inference based on the usage of the geometrical properties of the underlying data space. Examples of such data types that especially lie on Riemannian manifolds, include directions of points on a sphere \citep{FiLeEm:1987, MaJu:1999}, shapes of configurations extracted from images \citep{BhBh:2012}, data sitting on Stiefel and Grassmann manifolds \citep{Chikuse:1999}, symmetric positive definite matrices arising as observations in diffusion tensor magnetic resonance imaging (DT-MRI) \citep{ZhChIb:2009, YuZhLiMa:2012}, and other types of medical images.

In common with the traditional Euclidean case, statistical inference on the aforementioned manifolds begins by defining the notion of the mean on a certain metric space where data resides. Suppose a random object $\bX$ is defined on a metric space $(\mathcal{M},\rho)$, and let $\mathcal{Q}(\cdot)$ be the probability measure of $\bX$. Then one may consider adopting the traditional definition of the mean to generalize the notion of the mean on an arbitrary metric space, i.e., $\mathbb{E}(\bX) = \int_{\mathcal{M}} \bx \mathcal{Q}(d\bx)$. Unfortunately, however, this attempt at generalization is not directly applicable to the non-Euclidean setting, because it contains non-vector valued integral, which appears to be analytically intractable. Therefore, the conventional definition of the mean needs to be adapted so that it can go beyond Euclidean spaces. The most commonly used one in the literature is the Fr\'echet mean \citep{Fr:1948}, in which the mean is defined as a minimizer of the real valued function defined on metric spaces. To be more specific, for any $\bq \in \mathcal{M}$, consider the Fr\'echet function of the following form
\begin{align}
\mathcal{F} : &\ \mathcal{M} \rightarrow \mathbb{R} \notag\\
&\ \bq \in \mathcal{M} \mapsto \mathcal F(\bq) : \mathbb{E}\left( \rho^2 (\bX,\bq)\right) = \int_\mathcal{M} \rho^2(\bx,\bq) \mathcal{Q}(d\bx),
\end{align}
where $\rho$ denotes generic metric on $\mathcal{M}$. Then the Fr\'echet mean is defined as the minimizer of the Fr\'echet function above, i.e.,
\begin{align}
    \bmu_F = \argmin_{\bq\in\mathcal{M}}\int_\mathcal{M} \rho^2(\bx,\bq) \mathcal{Q}(d\bx). \label{eq:Frechet_mean}
\end{align}
And also, for a given observation $\bx_1,\cdots,\bx_n \in \mathcal{M}$, which consists $n$ independent realizations of $\bX$, the sample Fr\'echet mean is defined as $\overline{\bX} = \argmin_{\bq \in \mathcal{M}}\sum_{i=1}^n\rho^2(\bx_i,\bq)$. Given the relation between the mean and the variance, the above generalization of the mean makes intuitive sense, because it gives analogous definition of the Euclidean mean which is characterized as the minimizer of the variance function, the sum of the squared deviation. In this regard, the Fr\'echet function itself is commonly referred to as the Fr\'echet variance. 

But, first and foremost, what needs to be emphasized is that a metric $\rho$ is not unique for any particular manifold, and there are many possible choices. Regarding this issue, two different types of distance functions have been typically considered in the literature of manifold valued data analysis. The first possible choice is the intrinsic distance, that is the geodesic distance associated with the Riemannian structure $\mathrm{\mathbf{g}}$ on $\mathcal{M}$. The other type of distance is the Euclidean distance induced by the embedding $J :\mathcal{M} \rightarrow E^d$, which is also referred to as the extrinsic distance. The former and the latter distances lead to the intrinsic and the extrinsic data analysis, respectively.

Most of the previous works on analyzing manifold valued data have mainly focused on developing statistical methods based on variants of the Fr\'echet mean. For instance, by introducing the conditional Fr\'echet mean, \citet{PeMu:2019} developed a regression model having a random object in a metric space as a response variable. However, it is well known that the least squares based methods are severely degraded when there exists outliers in the data or the underlying data distribution is heavy tailed. Thus, the lack of statistical robustness, incurred by the squared distance involved in \eqref{eq:Frechet_mean}, becomes apparent in the Fr\'echet mean as well. Whereas, in the Euclidean space setting, considerable efforts have been devoted to improving the robustness of estimators\citep[see][and references therein for a review]{Hu:1964, HuRo:2009, HaRoRoSt:1986}, far less attention has been paid to manifolds. Indeed, one simple way to enhance robustness of estimators is replacing the squared distance by the unsquared distance. In the case of Euclidean space, where $\mathcal{M} = \mathbb{R}^d, \rho(\bx,\bx^\prime) = \Vert \bx - \bx^\prime \Vert$, this approach has a geometric median \citep{HA:1948} as a special case. Along the same line, the intrinsic geometric median on Riemannian manifolds, obtained by minimizing the Fr\'echet function associated with the unsquared geodesic distance, has been proposed by \citet{FlVeJo:2009}. Motivated by the success of the intrinsic geometric median, the primary contribution of this paper is to develop a novel robust location parameter within an extrinsic framework, which entails a computationally efficient algorithm. Moreover, adopting the concept of the classical local polynomial modeling, we implement the robust extrinsic local regression model for manifold valued response and Euclidean predictor. This can be accomplished by extending the concept of the proposed extrinsic median to the notion of the conditional extrinsic median.

The rest of this paper is organized as follows. The proposed extrinsic median is introduced in \cref{sec:2}, along with a brief review of the extrinsic framework for manifold data analysis. Application of the extrinsic median to two different manifolds is also demonstrated in \cref{sec:3}. In \cref{sec:Robust_Reg}, we develop the robust extrinsic local regression (RELR) model, with algorithmic details. In \cref{sec:simulation}, the proposed RELR is implemented in the Kendall's planar shape space and its promising properties are illustrated through simulation studies. Finally, we conclude the paper in \cref{sec:conclusion} with a short discussion and possible directions for the future study.


\section{Extrinsic Median} \label{sec:2}

In this section, we develop the extrinsic median which provides a statistically robust and computationally efficient way of estimating the center point of the data residing on manifolds. Before describing our method, it is useful to begin with a brief review of the extrinsic framework for manifold valued data analysis on which the proposed scheme is based, and the motivation that initiated this study.

\subsection{Extrinsic framework on manifold valued data analysis} \label{sec2: EDA}
The essential idea behind the extrinsic analysis is that any $d$-dimensional manifold $\mathcal{M}$ can be embedded in a higher-dimensional Euclidean space $\mathbb{R}^D$, where $d < D$ \citep{Wh:1944} via an embedding $J$. Thus, to understand the extrinsic approach, it is necessary to recall the definition of embedding $J$. First consider a differentiable map $J : \mathcal{M} \rightarrow \mathbb{R}^D$, whose differential $d_{\bp}J : T_{\bp}\mathcal{M} \rightarrow T_{J({\bp})}\mathbb{R}^D$ is a one-to-one, where $T_{\bp}\mathcal{M}$ and  $T_{J({\bp})}\mathbb{R}^D$ denote the tangent space at $\bp \in \mathcal{M}$ and the tangent space at $J(\bp)$ on $\mathbb{R}^D$, respectively. Note that the class of differentiable maps specified above is called an immersion. Then the one-to-one immersion is called the embedding if it is a homeomorphism from $\mathcal{M}$ to $J(\mathcal{M})$ with the induced topology. Also of note is that the embedding is unfortunately not unique in general, and not all choices of embedding lead to a good estimation result. In this context, the extrinsic approach has been carried out under the premise that the selected embedding preserves intrinsic geometry of the original manifold. Therefore, the embedding satisfying the following condition is typically preferred within extrinsic framework. 
For a Lie group $G$ acting on $\mathcal{M}$, the embedding $J : \mathcal{M} \rightarrow \mathbb{R}^D $ is referred to as the $G$ equivariant embedding if there exists the group homomorphism
$\phi : G \rightarrow \operatorname{GL}_D(\mathbb{R} \ \text{or} \ \mathbb{C})$ satisfying $J(g \bp) = \phi(g) J(\bp), \forall \bp \in \mathcal{M}, g \in G$, where $\operatorname{GL}_D(\mathbb{R} \ \text{or} \ \mathbb{C})$ denotes the general linear group which is the group of $D\times D$ invertible real, or complex matrices. This definition indicates that the group action of $G$ can be recovered in the embedded space $J(\mathcal{M})$ through $\phi$. Therefore, in light of the above, a great amount of geometric feature of the manifold is preserved in the embedded Euclidean space via the equivariant embedding. And the extrinsic distance between two points a manifold can be straightforwardly computed in an embedded space via the Euclidean norm.

Considering all the notions described above, the extrinsic mean $\bmu_{E}$ on a manifold is defined as the minimizer of the Fr\'echet function associated with the extrinsic distance via an embedding $J : \mathcal{M} \rightarrow \mathbb{R}^D$
\begin{align}
    \bmu_E = \argmin_{\bq\in\mathcal{M}}\int_\mathcal{M} \Vert J(\bx)-J(\bq) \Vert^2 \mathcal{Q}(d\bx) \label{eq:extrnsic_mean}.
\end{align}
As compared to the intrinsic mean, the use of the extrinsic approach has several advantages, including (1) computational efficiency \citep{BhElLi:2011}, (2) milder conditions for existence and uniqueness of the solution.
Moreover, the sample extrinsic mean often has a closed form solution. Thus, we here derive the extrinsic mean in an explicit form. To do this the following definition which gives the uniqueness condition of the extrinsic mean should be noted first. A point $\by \in \mathbb{R}^D$ is said to be $J$-nonfocal if there exists a unique point $\bp \in \mathcal{M}$ satisfying $\inf_{\bx \in \mathcal{M}} \Vert \by - J(\bx) \Vert = \Vert \by - J(\bp) \Vert$. Then we let $\bmu = \int_{\mathbb{R}^D} \bu \widetilde{\mathcal{Q}}(d\bu)$ be the mean vector of the induced probability measure $\widetilde{\mathcal{Q}} = \mathcal{Q} \circ J^{-1}$, which is the image of $\mathcal{Q}$ in $\mathbb{R}^D$. Then the Fr\'echet function associated with the extrinsic distance, the right hand side of \eqref{eq:extrnsic_mean}, also can be written as
\begin{align}
    \mathcal{F}(\bq) = \Vert J(\bq) - \bmu \Vert^2 + \int_{\mathbb{R}^D} \Vert \bx - \bmu \Vert^2 \widetilde{\mathcal{Q}}(d\bx).\label{eq:sup_frechet}
\end{align}
Hence, we have $\inf_{\bq \in \mathcal{M}} \mathcal{F}(\bq) = \inf_{J(\bq) \in \widetilde{\mathcal{M}}}\Vert J(\bq) - \bmu \Vert^2 + \int_{\mathbb{R}^D} \Vert \bx - \bmu \Vert^2 \widetilde{\mathcal{Q}}(d\bx)$, where $\widetilde{\mathcal{M}} = J(\mathcal{M})$ denotes the image of the embedding. This indicates the set of points $\bx \in \mathcal{M}$ satisfying $\inf_{j(\bq) \in \widetilde{\mathcal{M}}} \Vert J(\bq) - \bmu \Vert = \Vert J(\bx) - \bmu \Vert$ consists the extrinsic mean set. And since, \eqref{eq:sup_frechet} is minimized on $\widetilde{\mathcal{M}}$ by $J(\bq) = \mathcal{P}(\bmu)$, where $\mathcal{P}:\mathbb{R}^D \rightarrow \widetilde{\mathcal{M}}$ such that for $\forall \bu^\prime \in \widetilde{\mathcal{M}}$, $\mathcal{P}(\by)=\{ \bu \in \widetilde{\mathcal{M}} : \Vert \bu - \by \Vert \leq \Vert \bu^\prime - \by \Vert \}$, the extrinsic mean uniquely exists if and only if the mean vector $\bmu$ is a $J$-nonfocal point. In that case, the extrinsic mean is obtained by taking the inverse of the embedding, i.e., $\bmu_E = J^{-1} ( \mathcal{P}(\bmu))$. Following from the above, the sample extrinsic mean is obtained in a straightforward manner. Suppose we observe $\bx_1, \dots, \bx_n \in \mathcal{M}$, consisting of independent and identically distributed copies of $\bX$, then the sample extrinsic mean  is given by $\overline{\bX}_E = J^{-1} \{ \mathcal{P} ( \overline{J(\bX)})\},$ where $\overline{J(\bX)} = \sum_{i=1}^n J(\bx_i)/n$. Theoretical properties of the sample extrinsic mean, including asymptotic distribution, consistency, and the uniqueness conditions are well established in \citet{BhPa:2003,BhPa:2005}.

\subsection{Extrinsic Median} \label{sec2: EMed}

Before proceeding to present our proposed method, we begin by giving a quick overview of the existing Euclidean geometric median. In the Euclidean multivariate setting, a large body of research has been devoted to developing the robust estimation of the central point \citep[see for a review][]{Sm:1990}, among which the geometric median, initially proposed by \citet{HA:1948}, has received the greatest attention over the last decades due both to its nice robustness properties and computational efficiency \citep{Ca:2013,Ca:2017}.
The geometric median of a random variable $\bX \in \mathbb{R}^k$ is defined by $\bm = \argmin_{\bq \in \mathbb{R}^k} \mathbb{E} \Vert \bX - \bq \Vert$, or alternatively but equivalently, is obtained by minimizing $\int_{\mathbb{R}^k} \left( \Vert \bx-\bq \Vert - \Vert \bx \Vert \right)\mathcal{Q}(d\bx).
$
Note that the latter expression has been more commonly adopted in practice, since no assumption regarding the first order moment of $\bX$ needs to be imposed. Moreover, when $k=1$ the above definition corresponds with the classical notion of the median which is defined in terms of the cumulative distribution function. In this sense, the geometric median plays a role of the multivariate generalization of the univariate median. Now suppose that we observe $\mathcal{X}=\{\bx_1, \cdots, \bx_n\}$ consisting of $n$ independent and identically distributed realizations of $\bX$, then the sample geometric median $\widehat{\bm}$, which provides the natural estimation of $\bm$, is obtained by finding the optimal value that minimizes the
sum of Euclidean distances to given data points, i.e, 
\begin{align*}
\widehat{\bm} = \argmin_{\bq \in \mathbb{R}^k} \sum_{i=1}^n \left( \Vert \bx_i-\bq \Vert - \Vert \bx_i \Vert \right).
\end{align*}
The above optimization problem is also known as the Fermat-Weber problem \citep{We:1929}, and the numerical algorithm for solving the geometric median problem was firstly introduced by \citet{We:1937}. It is shown in \citet{Ke:1987} that the sample geometric median is uniquely determined unless all the given observations do not lie on the same line.

Although many nice properties have been investigated including invariance under rotation and translation, asymptotic behavior \citep{MoNoOj:2010}, concentration \citep{Mi:2015}, however, the most notable advantage of the geometric median over the mean is that it provides a robust estimation of the centrality under the presence of noise in the data. The robustness of the estimator is usually measured by the breakdown point. For a given data $\mathcal{X}$ in the above, we further consider the outlier-contaminated data $\mathcal{X}^\ast_m = \{\bx_1^\ast,\cdots,\bx_m^\ast,\bx_{m+1},\cdots,\bx_{n}\}$, where the first $m$ elements are replaced by extreme noises, then the breakdown point of the estimator $T_n$ is defined by 
\begin{align*}
B(T_n) = \min_{1 \leq m\leq n} \left\{ \frac{m}{n} \ \left\vert \ \sup_{\mathcal{X}^\ast_m} \Vert T_n(\mathcal{X}^\ast_m) - T_n(\mathcal{X}) \Vert = \infty \right. \right\}.
\end{align*}
In an intuitive sense, the breakdown point can be interpreted as the highest proportion of contamination that the estimator can tolerate before the difference between the estimated result obtained from the contaminated data and the initial result goes to infinity. Being less affected by outliers, the geometric median achieves the asymptotic breakdown point of $0.5$ \citep{LoRo:1991}. This indicates that the geometric median can provide a good estimation result even though up to half of the data is corrupted. Note that the breakdown point of the sample mean is $1/n$, meaning that only one single extreme value changes the estimation result arbitrary.

We now turn our attention to manifolds. Much of the research regarding estimation of the central location of data on manifolds has focused on the variants of the Fr\'echet mean including the intrinsic, extrinsic mean. However, the common drawback of the least square based methods is their lack of robustness to extreme values, which makes the Fr\'echet mean inevitably sensitive to heavy tailed distributions and outlying values. Nevertheless, even though there has been a considerable increase in the need for robust statistical methods on manifolds, less has been done on this issue. The pioneering attempt to address this is seen in the work of \citet{FlVeJo:2009}, where they proposed the intrinsic median by substituting the squared geodesic distance employed in the Fr\'echet function with the unsquared one, i.e., $\argmin_{\bq \in \mathcal{M}} \int_{\mathcal{M}} \rho(\bx,\bq) \mathcal{Q}(d\bx)$. Although their approach has had a great deal of success in generalizing the notion of the median to Riemannian manifolds by attaining the same breakdown point as in the Euclidean case, it has some inherent drawbacks that may limit its application. (1) For example, it is often difficult to derive conditions for the uniqueness and the existence of the intrinsic median without restrictions on its support. (2) Moreover, even when the intrinsic median exists, it requires iterated algorithms on manifolds which may incur a large amount of computational overhead. These drawbacks highlight the need for the development of novel approaches aimed at giving a more computationally efficient method in which the existence and uniqueness conditions are well established and easy to understand.

In an attempt to address the methodological shortcomings of the intrinsic median described above, we propose the following new robust location parameter by making use of the unsquared extrinsic distance,
\begin{align}
    \bold{m}_E = \argmin_{\bq\in\mathcal{M}}\int_\mathcal{M} \Vert J(\bx)-J(\bq) \Vert \mathcal{Q}(d\bx) \label{eq:Extrinsic_med}.
\end{align}
Given observations $\bx_1,\cdots, \bx_n$ consisting of independent and identically distributed copies of manifold valued random variable $\bX$, the above location parameter, which we call the population extrinsic median, can be estimated by replacing $\mathcal{Q}$ with the empirical measure $\widehat{\mathcal{Q}} = 1/n\sum_{i=1}^n\delta_{\bx_i}$, i.e,
\begin{align}
    \widehat{\mathbf{m}}_E &= \argmin_{\bq\in\mathcal{M}}\sum_{i=1}^n \Vert J(\bx_i)-J(\bq) \Vert \notag\\
    &= J^{-1} \left( \mathcal{P} \Big(\argmin_{\bm \in\mathbb{R}^D}\sum_{i=1}^n \Vert J(\bx_i)- \bm \Vert \Big) \right). \label{eq:Sample_Ext_Med}
\end{align}
Unlike the sample extrinsic mean which has a closed form expression depending on the projection map, the sample extrinsic median requires an iterative algorithm, called the Weiszfeld's algorithm for solving the inner minimization problem. But taking advantage of Euclidean geometry, the proposed extrinsic approach allows us to exploit the original form of Weiszfeld algorithm without requiring any further modifications. Indeed, the following \Cref{al:Ext_median} for solving \eqref{eq:Sample_Ext_Med} can be easily derived due to the convexity of the object function associated with the Euclidean norm, $f(\pbm) = \sum_{i=1}^n \Vert J(\bx_i)- \pbm \Vert$.
\begin{algorithm}
\caption{Extrinsic Median}
        \begin{algorithmic}[1]
            \Input{$n$ observations $\mathcal{X} = \{\bx_1, \cdots , \bx_n\}$}
            \Initialize{$t=0, \pbm^0 \ \text{and} \ \varepsilon$}
            \While{$\Vert \pbm^{t+1}-\pbm^{t}\Vert < \varepsilon$}  
                \State Compute the gradient direction $\nabla f(\pbm^{t})$
                    \begin{align*}
                        \sum_{i=1}^n \frac{\pbm^{t}-J(\bx_i)}{\Vert \pbm^{t}-J(\bx_i)\Vert}
                    \end{align*}
                \State Compute the step size
                    \begin{align*}
                       s^{t} =  \left(\sum_{i=1}^n \frac{1}{\Vert \pbm^{t} - J(\bx_i) \Vert }\right)^{-1}
                    \end{align*}
                \State Update $m^{t+1}$
                     \begin{align*}
                        \pbm^{t+1} = \pbm^{t} - s^t \cdot \nabla f(\pbm^{t})
                    \end{align*}
                \State{$t \leftarrow t+1$}
            \EndWhile \State \textbf{end while}
            \Output{Estimated robust estimator $\widehat{\mathbf{m}}_E = J^{-1}(\mathcal{P}(\pbm^\ast))$} \Comment{$\pbm^\ast$ denotes the optimal value.}
        \end{algorithmic}\label{al:Ext_median}
\end{algorithm}

In fact, when incorporated into the extrinsic framework, without incurring the computational overhead encountered in the Riemannian manifolds optimization, our approach possess a practical advantage over the intrinsic geometric median algorithm. Specifically, in contrast with the intrinsic geometric median algorithm \citep{FlVeJo:2009}, 
\begin{align*}
    \pbm^{t+1} = \texttt{Exp}_{\pbm^t}(\alpha \mathbf{v}^t), \ \mathbf{v}^t = \sum_{i=1}^n \frac{\texttt{Log}_{\pbm^t}(\bx_i)}{\rho(\pbm^t,\bx_i)} \cdot \left(\sum_{i=1}^n \frac{1}{\rho(\pbm^t,\bx_i)} \right)^{-1},
\end{align*}
in which $\texttt{Exp} : T_{\bm^t} \rightarrow \mathcal{M}$, $\texttt{Log} : \mathcal{M} \rightarrow T_{\bm^t} $ have to be repeatedly evaluated at each iteration, our method can reduce the additional computational cost caused by the above exponential and logarithm mapping. As indicated above, it should be emphasized that although the data lie on manifolds, the proposed algorithm itself operates in Euclidean space without suffering from geometrical restrictions and constraints posed by non-Euclidean data domains.

Also of importance is that when all given data points are not colinear (i.e., there doesn't exist $\by,\bold{z} \in \mathbb{R}^p$ and $\alpha_1 , \cdots \alpha_n \in \mathbb{R}$, such that $\forall i=1,\cdots,n$, $\bx_i = \by + \alpha_i \bold{z}$), the Weiszfeld's algorithm converges always to the unique optimal solution \citep{Ku:1973}.
 For each iteration step, $\pbm^t \not\in \{\bx_1, \cdots, \bx_n\}$ is typically assumed in order to ensure that the proposed algorithm converges to the global optimal solution. Details of the algorithm including derivation and their convergence analysis are deferred to \Cref{sec:Robust_Reg}, in which we discuss the algorithm for solving the robust extrinsic local regression (\Cref{al:RELR}) of which \Cref{al:Ext_median} is a special case.

\section{Applications of Extrinsic Median}\label{sec:3}

In this section, the practical applicability and performance regarding robustness of the extrinsic median is examined through simulation studies on two important manifolds. To gain further insights into the extrinsic median, the experiment was carried out under different conditions which may possibly be encountered in practice. Results are compared with competing methods including the extrinsic mean.

\subsection{Unit sphere}\label{sec:app_unit_sphere}

The first and simplest application is an $d$-unit sphere, $\mathcal{S}^{d}=\{\bx\in\mathbb{R}^{d+1}:\Vert \bx \Vert = 1 \}$, which is a $d$-dimensional submanifold of $\mathbb{R}^{d+1}$. It can be embedded into $\mathbb{R}^{d+1}$ through the inclusion map $\iota: \mathcal{S}^{d} \rightarrow \mathbb{R}^{d+1}$, $\iota(\bx)=\bx$. The projection map $\mathcal{P} : \mathbb{R}^{d+1} \rightarrow \mathcal{S}^d $ is defined by $\mathcal{P}(\bmu) = \bmu/\Vert \bmu \Vert$, where $\bmu = \int_{\mathbb{R}^{d+1}} \bx \widetilde{\mathcal{Q}}(d\bx)$ is the mean vector calculated in the ambient space of $\mathbb{R}^{{d+1}}$ and $\widetilde{\mathcal{Q}} = \mathcal{Q} \circ \iota^{-1}$ denotes the induced probability measure. Note that $\bmu$ is $\iota$-nonfocal unless $\bmu=\bf{0}$. For further details about statistical analysis on $\mathcal{S}^d$, we refer to \citet{FiLeEm:1987}, \citet{MaJu:1999} and references therein. 

In the following, the performance of extrinsic median on $\mathcal{S}^{d}$ is illustrated by simulation studies. To ease visualization of how the generated data looks like and how the extrinsic median is capable to provide robust estimation than the extrinsic mean, the simplest case $\mathcal{S}^1$, for which data is observed as the form of direction on a unit circle in $2$-dimensional Euclidean plane $\mathbb{R}^2$, is considered. Note that the data on $\mathcal{S}^1$ is typically represented by an angle measured in radians $\theta \in [0,2\pi)$, or the unit vector $\bx = (\cos\theta,\sin\theta)^\top$ from the origin. The performance of the extrinsic median is compared under two different simulation scenarios as follows. In the first scenario, outliers are artificially imposed to the von Mises ($\operatorname{VM}$) distribution, whereas in the second scenario, heavy tailed random observations are generated from the general wrapped $\alpha$ stable ($\operatorname{WS}$) distribution. The detailed description of each scenario is given in the following.

\noindent\textbf{Scenario 1) :  The von Mises distribution with the Normal outliers.}\\
\indent We suppose $\theta$ follows the von Misese distribution, $\operatorname{VM} (\mu,\kappa)$ with the density function 
\begin{align*}
    f(\theta) = \frac{e^{\kappa \cos(\theta-\mu)}}{2\pi I_0(\kappa)},
\end{align*}
where $\mu, \kappa$ denote the mean direction and concentration parameter, respectively and $I_0(\kappa)$ is the modified Bessel function of order $0$. Note that larger value of $\kappa$ means higher concentration towards $\mu$. We first generate random data $\{\theta_i\}_{i=1}^n$ consisting independent and identically distributed copies of $\theta \sim \operatorname{VM} (\mu,\kappa)$, then $n_{\text{cont}}$ outliers $o_j^\prime \stackrel{iid}{\sim} \operatorname{Normal}(\mu_\textbf{out},\sigma^2)$, where $\mu_\textbf{out} \neq \mu$, are added to the initial data set so that the contamination level satisfies the prespecified value $r = n_{\text{cont}}/n$. Additionally normalization of the generated outliers $o_j = o_j^\prime \hspace{-.2cm}\pmod{2\pi}$ is required to ensure $0 \leq o_j < 2\pi$.

\noindent\textbf{Scenario 2) : The wrapped $\alpha$-stable random variable.}\\
\indent The density function of a wrapped $\alpha$-stable random variable $\theta$ is given by
\begin{align}
    f(\theta) = \frac{1}{2\pi} + \frac{1}{\pi}\sum_{k=1}^\infty \exp(-\tau^\alpha k^\alpha) \cos \left( k(\theta-\mu) - \tau^\alpha k^\alpha \beta \tan\frac{\alpha \pi}{2}\right),\label{eq:WS_dist}
\end{align}
where $0 < \alpha \leq 2$, $\tau \geq 0$ and $\vert \beta \vert \leq 1$ denote the shape, dispersion and skewness parameters, respectively. Note that small values of $\alpha$ yield heavy tailed distributions but larger $\tau$ values yield more highly dispersed distributions. The benefit of using wrapped $\alpha$-stable distribution is that it provides a high degree of flexibility in modeling directional data in the sense that it contains many popular circular distributions as special cases, including the wrapped normal distribution ($\alpha=2$) and the wrapped Cauchy distribution ($\alpha=1, \beta=0$); see \citet{JaSe:2001} for further details. 

Representative illustrations of simulation scenario 1 and scenario 2 are
displayed in \Cref{fig:circle} (left and right panel, respectively), together with estimated values. In both scenarios, we observed that extrinsic mean estimations were forced by outliers to be pulled far away from the true mean direction ($\mu=0$, i.e., $\bx = (1,0)$ in the Cartesian coordinate system). In \Cref{tab:circle_scenarios}, the extrinsic median is compared with the extrinsic mean in terms of norm of difference between the true mean direction and the estimated direction. The results are averaged over 20 replications. In the first scenario, four different settings are considered according to the level of contamination, $r \in \{0, 0.1, 0.2, 0.4\}$, where 0 represents no outlier exists. As would be expected, the result obtained from the first scenario indicates that as the contamination level becomes higher, the extrinsic mean is far more vulnerable to the presence of outliers than the extrinsic median. The bottom panel of the table shows the result of the scenario 2 in which we
fix $\beta=0$ for the symmetry of the distribution and vary the tail heaviness level by adjusting $\alpha$ from $0.1$ to $2$, and the dispersion of the data is controlled by differing $\tau = 0.2, 2$. It is observed that extrinsic median not only has a better predictive ability in the case of heavy tailed data which corresponds to small values of $\alpha$, but also a comparable performance was achieved even in non-heavy tailed data, generated from the wrapped normal distribution ($\alpha = 2$).
\begin{figure}
\centering
    \includegraphics[width=0.49\linewidth]{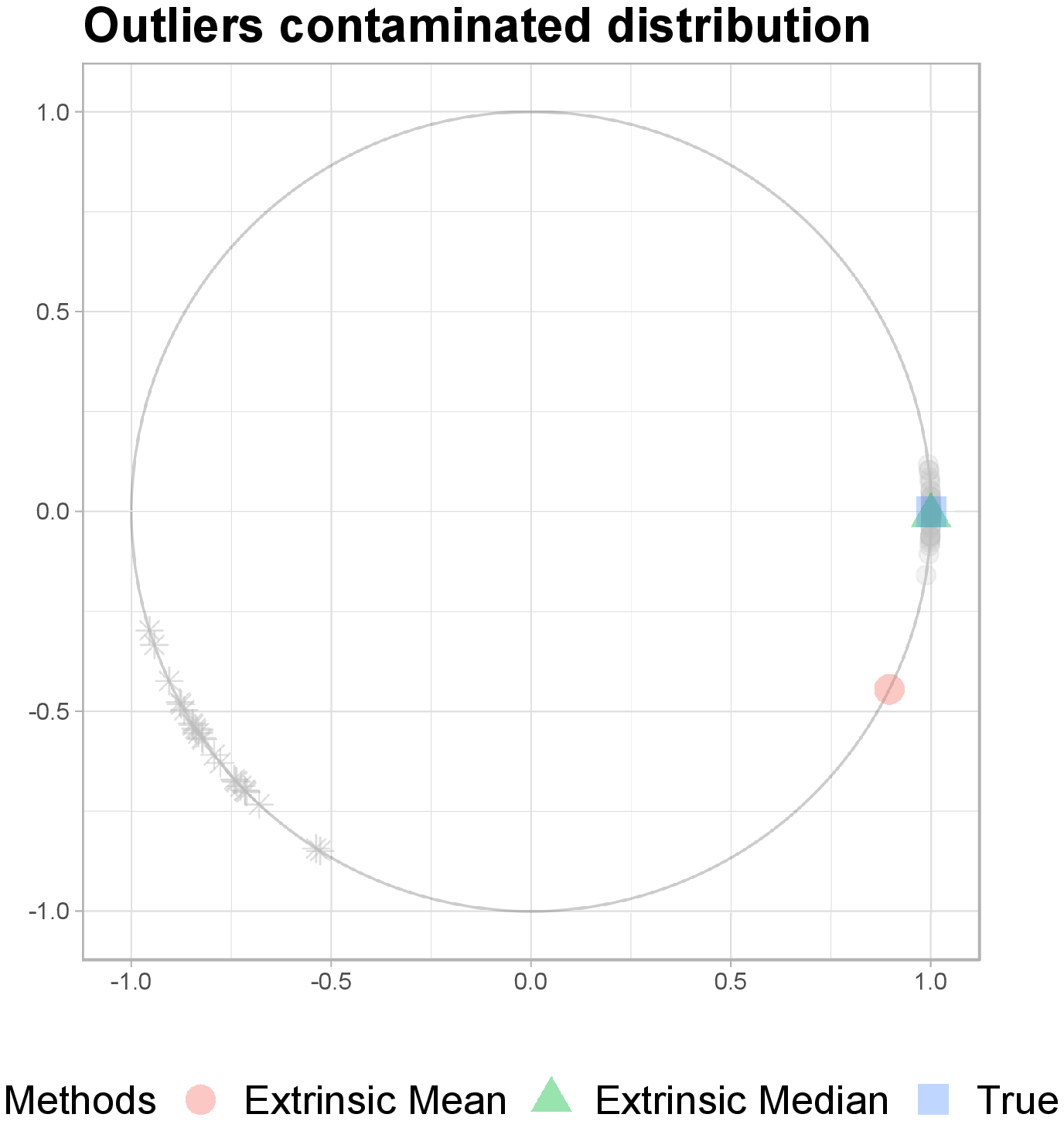}
    \includegraphics[width=0.49\linewidth]{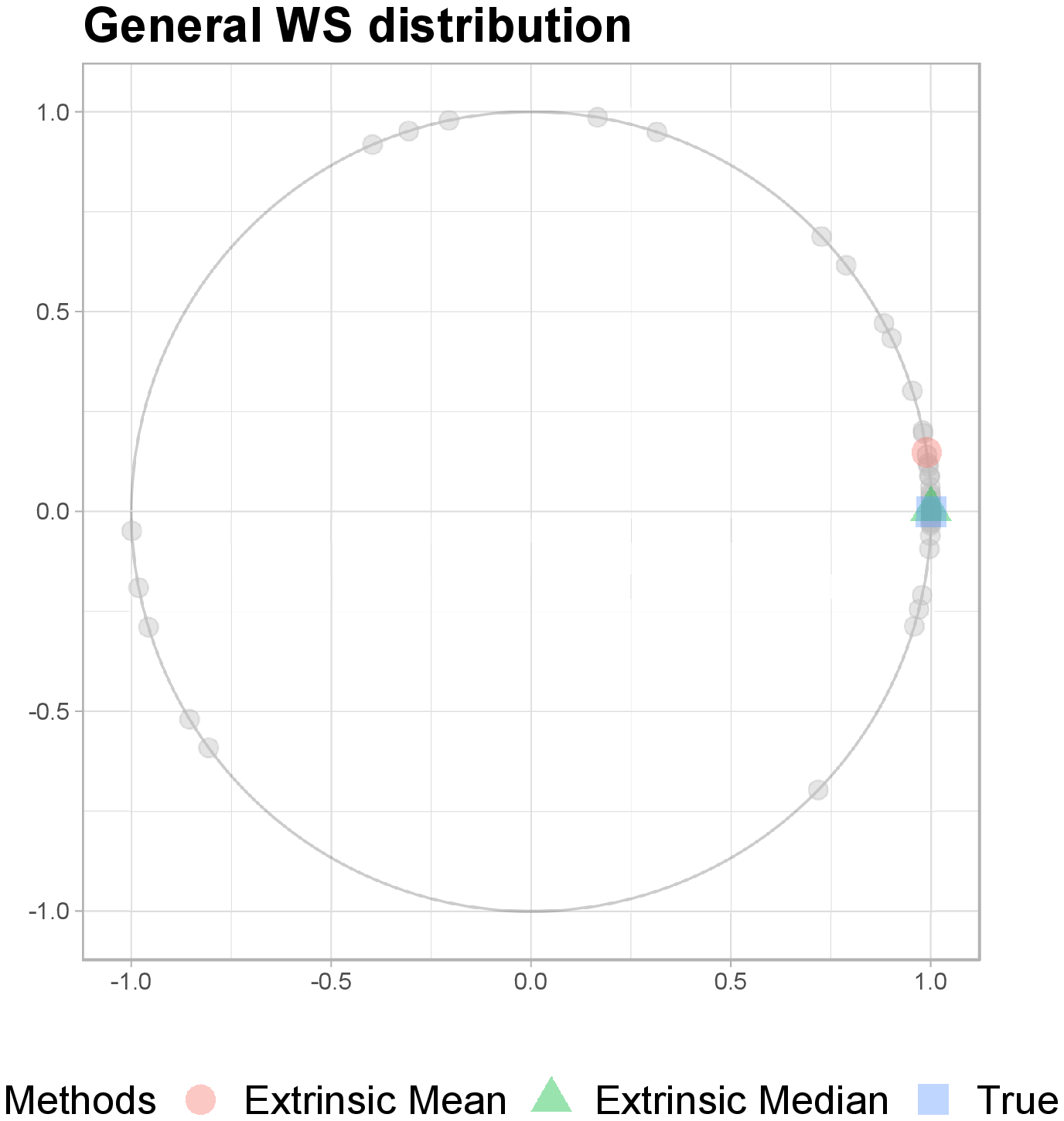}
    \caption{Left: an example of the Scenario 1), consisting of observations drawn from the VM distribution and outliers displayed in light grey circles and asterisks, respectively. Right: Scenario 2). In both setting, the fit of the extrinsic mean and median are displayed with the true mean value.}
    \label{fig:circle}
\end{figure}

\begin{table}[ht]
    \centering
    \begin{tabular}{lllllllllllllll}
    \multicolumn{9}{l}{\textbf{Scenario 1 :} The outlier contaminated case}\\
    \toprule
        & \multicolumn{8}{c}{Ratio of outliler}\\
        \cmidrule(lr){2-9}
        &  \multicolumn{2}{c}{No outlier} & \multicolumn{2}{c}{0.1} & \multicolumn{2}{c}{0.2} & \multicolumn{2}{c}{0.4} \\
         \cmidrule(lr){2-3}  \cmidrule(lr){4-5} \cmidrule(lr){6-7} \cmidrule(lr){8-9}
         $N$ & E. Mean & E. Med & E. Mean & E. Med & E. Mean & E. Med & E. Mean & E. Med \\
         \midrule  
        10 & 0.0100 & 0.0090  & 0.0550 & 0.0129 & 0.1376 & 0.0132 & 0.2778 & 0.0169\\
        50 & 0.0064 & 0.0066 & 0.0663 & 0.0106 & 0.1311 & 0.0115 &  0.2702 & 0.0122\\
        100 & 0.0032 & 0.0033 & 0.0624 & 0.0101 & 0.1315 & 0.0087 & 0.2732 & 0.0118\\
        200 & 0.0025 & 0.0028 & 0.0655 & 0.0099 & 0.1331 & 0.0101 & 0.2730 & 0.0136\\
        \bottomrule
    \end{tabular}\\ \vspace{.5cm}
    \begin{tabular}{llllllllllllllll}
    \multicolumn{9}{l}{\textbf{Scenario 2 :} The heavy tailed distribution }\\
    \toprule
         & & \multicolumn{8}{c}{$\alpha$}\\
        \cmidrule(lr){3-10}
        & &  \multicolumn{2}{c}{0.1} & \multicolumn{2}{c}{0.5} &  \multicolumn{2}{c}{1} & \multicolumn{2}{c}{2} \\
         \cmidrule(lr){3-4}  \cmidrule(lr){5-6} \cmidrule(lr){7-8} \cmidrule(lr){9-10}
        $\tau$ &$N$ & E. Mean & E. Med & E. Mean & E. Med & E. Mean & E. Med & E. Mean & E. Med \\
         \midrule  
        \multirow{4}{*}{0.2} &  10 & 0.3411 & 0.1658 & 0.1971 & 0.0540 &  0.1044 & 0.0505 & 0.0376 & 0.0393\\
        &50 & 0.1053 & 0.0011 & 0.0709 & 0.0135 & 0.0454 & 0.0199 & 0.0188 & 0.0197\\
        &100 & 0.0653 & 0.0003 & 0.0432 & 0.0095 & 0.0364 & 0.0141 & 0.0123 & 0.0194\\
        &200 & 0.0582 & 0.0003 & 0.0211 & 0.0062 & 0.0198 & 0.0091 &  0.0086 & 0.0095\\
         \midrule  
        \multirow{4}{*}{2} &10 & 0.3105 & 0.1294 & 0.4291 & 0.4054 & 0.4446 & 0.3963 & 0.3876 & 0.4501 \\
        &50 & 0.1471 & 0.0014 & 0.2221 & 0.1112 & 0.1619 & 0.1550 &  0.1824 & 0.2097\\
        &100 & 0.1312 & 0.0023 & 0.1434 & 0.0858 & 0.1415 & 0.1339 &  0.2075 & 0.2641 \\
        &200 & 0.0867 & 0.0005 & 0.0894 & 0.0567 & 0.0967 & 0.0845 & 0.1161 & 0.1421\\
        \bottomrule 
    \end{tabular}
    \caption{Results of the experiment described in \Cref{sec:app_unit_sphere}. The result of scenario 1 and 2 are presented in the top and bottom panels of the table, respectively}
    \label{tab:circle_scenarios}
\end{table}
\subsection{Planar Shape}\label{sec:app_planarshape}

For the second application of the extrinsic median, we consider the Kendall's planar shape space of $k$-ads, denoted by $\Sigma_2^k$ \citep{Kendall:1984} which is the most popular manifold in landmark based shape analysis literature. Before proceeding to present simulation study on the planar shape space, we give necessary preliminaries about this space.

The planar shape can be defined as a random object that is invariant under the Euclidean similarity transformation. Therefore, the planar shape is identified as the remaining geometric information after filtering out the effect of translation, scaling, and rotation. To ease understanding of this nonlinear manifold, let us begin by demonstrating the geometry of the planar shape space. First, the unregistered $k$-ads which is a landmark configuration that describes a shape of an object can be conveniently placed on a complex plane as a set of $k$ complex numbers, i.e., $\bz = (z_1,\cdots,z_k)$, where $z_j = x_j + i y_j \in \mathbb{C}$. Then one can obtain the preshape of $\bz$ by quotienting out the effect of translation and scale
\begin{align*}
    \bu = \frac{\bz-\langle \bz \rangle}{\Vert \bz-\langle \bz \rangle \Vert},
\end{align*}
where $\langle \bz \rangle = (\bar{z},\cdots,\bar{z})$, and $\bar{z} = \frac{1}{k}\sum_{j=1}^k z_j$. This indicates that the preshape space is equivalent to a complex hypersphere, $\mathbb{C}S^{k-1} = \left\{ \bu \in \mathbb{C}^k \vert \sum_{i=1}^k \bu_j = 0, \Vert \bu \Vert = 1 \right\}$. Then the shape $[\bz]$ of $\bz$ which is the geometric object that is invariant under a rotation effect, is obtained by considering all rotated version of $\bu$, i.e.,$[\bz] = \left\{ e^{i\theta} \bu : 0 \leq  \theta < 2\pi \right\}$.
As the shape is defined as the orbit of $\bu \in \mathbb{C}S^{k-1}$, Kendall's planar shape space $\Sigma_2^k =  \mathbb{C}S^{k-1} / SO(2)$ is the quotient space of the preshape space under the action of special orthogonal group of dimension $2$, $SO(2) = \{\bf{A} \in \operatorname{GL}_2 \vert \bf{A}^{-1} = \bf{A}^\top ,\operatorname{det}(\bf{A}) = 1\}$. Alternatively, the effects of scaling by a scalar $r>0$ and rotating by an angle $0 \leq \theta <2\pi$ can be simultaneously filtered out via multiplying by the complex number $\lambda= r e^{i\theta}$ from the centralized $k$-ad configuration $\bz-\langle \bz \rangle$, i.e., $[\bz] = \{ \lambda (\bz-\langle \bz \rangle) : \lambda \in \mathbb{C}\setminus\{0\} \}$. Due to this algebraically simpler characterization, the planar shape space is equivalently identified as the complex projective space $\Sigma_m^k \simeq \mathbb{C}P^{k-2}$ that is the space of all complex lines through the origin in $\mathbb{C}^{k-1}$.  
More detailed explanation of the geometrical structure of the shape manifold is provided in \citet{DrMa:1998,BhBh:2012}. 

We now describe the extrinsic approach in $\Sigma_2^k$. Due to \citet{Ke:1992}, in the Kendall's planar shape space the Veronese–Whitney embedding is typically used, which maps $\Sigma_2^k$ into the space of  $k \times k$ complex Hermitian matrices $\mathcal{S}(k,\mathbb{C})$ by
\begin{align}
    J :\ & \Sigma_2^k \rightarrow \mathcal{S}(k,\mathbb{C}) \notag\\
    &[\bz] \mapsto J([\bz]) = \bu \bu^\ast, \label{eq:VW_embedding}
\end{align}
where $\bu^\ast$ denotes the complex conjugate transpose of $\bu$. Furthermore, since $J(\bf{A} [\bz]) = \bf{A}\bu \bu^\ast \bf{A}^\ast$ holds for any ${\bf{A}} \in SU(k)$, where $SU(k) = \left\{ \bf{A} \in \operatorname{GL}_k(\mathbb{C})\ \vert \ \bf{AA}^\ast = \bf{I}, \det(\bf{A}) = 1 \right\}$ denotes the special unitary group, the Veronese-Whitney embedding is shown to be the $SU(k)$ equivariant embedding, i.e., $J ({\bf{A}} [\bz])= \phi({\bf{A}}) J([\bz])$. It follows directly by taking the Lie group homomorphism $\phi : \operatorname{SU}(k) \rightarrow \operatorname{GL}_k(\mathbb{C})$ such that $\phi({\bf{A}}){\bf{B}} = \bf{ABA}^\ast$, where ${\bf{B}} \in \mathcal{S}(k,\mathbb{C})$.
It also should be noted that the squared extrinsic distance of two planar shapes is defined in terms of the Frobenius norm of a complex matrix
\begin{align}
	\rho_E^2([\bz_1],[\bz_2]) &= \Vert J([\bz_1])-J([\bz_2]) \Vert_F^2 \notag\\
	&= \text{Trace} \Big(\left\{J([\bz_1])-J([\bz_2])\right\}\left\{J([\bz_1])-J([\bz_2])\right\}^\ast \Big)\notag\\
	&= \sum_{j=1}^k\sum_{i=1}^k \left\vert \{J([\bz_1])-J([\bz_2])\}_{i,j} \right\vert^2.
\end{align}
Since the above extrinsic distance takes into account every $k^2$ element of $J([\bz_1])-J([\bz_2])$, it can be viewed as the natural Euclidean distance between the two embedded shapes $J([\bz_1])$ and $J([\bz_2])$. Lastly, the inverse and projection map of the embedding, $J^{-1}(\mathcal{P}(\cdot))$ in \eqref{eq:Sample_Ext_Med}, remain to be identified. Let $\widetilde{\bX}$ be the arbitrary point on the ambient Euclidean space $\mathbb{R}^D$, then the projection mapping of $\widetilde{\bX}$ onto the image of the embedding $\widetilde{\mathcal{M}} = J(\Sigma_2^k)$ is given by $\bgamma \bgamma^\ast$, where $\bgamma$ is the unit eigenvector of $\widetilde{\bX}$ corresponding to the largest eigenvalue. Subsequently, the inverse map of $J^{-1}(\bgamma \bgamma^\ast) = [\bgamma]$ can be obtained directly from \eqref{eq:VW_embedding} without extra operations.

Now, in order to gauge the performance of the proposed method, we perform simulation experiments on the planar shape space by investigating the corpus callosum (CC) data extracted from the subset of ADHD-200 dataset (\url{http://fcon_1000.projects.nitrc.org/indi/adhd200/}). The original dataset includes functional magnetic resonance imaging (fMRI) scans of subjects categorized into four different groups based on their symptoms and conditions; (1) Typically developing children, (2) ADHD-Hyperactive, (3) ADHD-Inattentive, and (4) ADHD-Combined. The CC shapes of 647 subjects which consist of 50 landmarks were preprocessed and analyzed by \cite{HuStZh:2015} to illustrate their clustering method. In this experiment, however, only a subset of the data (the CC shapes extracted from 404 typically developing children) was utilized. Since the main aim of this simulation study is to see how the extrinsic median behaves robustly in a noisy environment, where a number of landmarks are contaminated by outliers, we further manipulated the data by assigning random noises generated from $\operatorname{Normal}(\mu=1000 , \sigma=5)$ to the real parts (the $x$ cooridinates) of the $10$th $\sim$ $15$th landmarks. The number of outliers were varied according to the noise level $r$ ranging from $0$ to $0.4$. Additionally, the extrinsic median was compared to several competing methods including the maximum likelihood estimator of the isotropic offset Gaussian distribution \citep{MaDr:1989, DrMa:1991} and different variants of the Fr\'echet means such as intrinsic mean, the Fr\'echet mean associated with the partial Procrustes distance and the extrinsic mean. 

\begin{figure}[ht!]
\centering
    \includegraphics[width=0.9\linewidth]{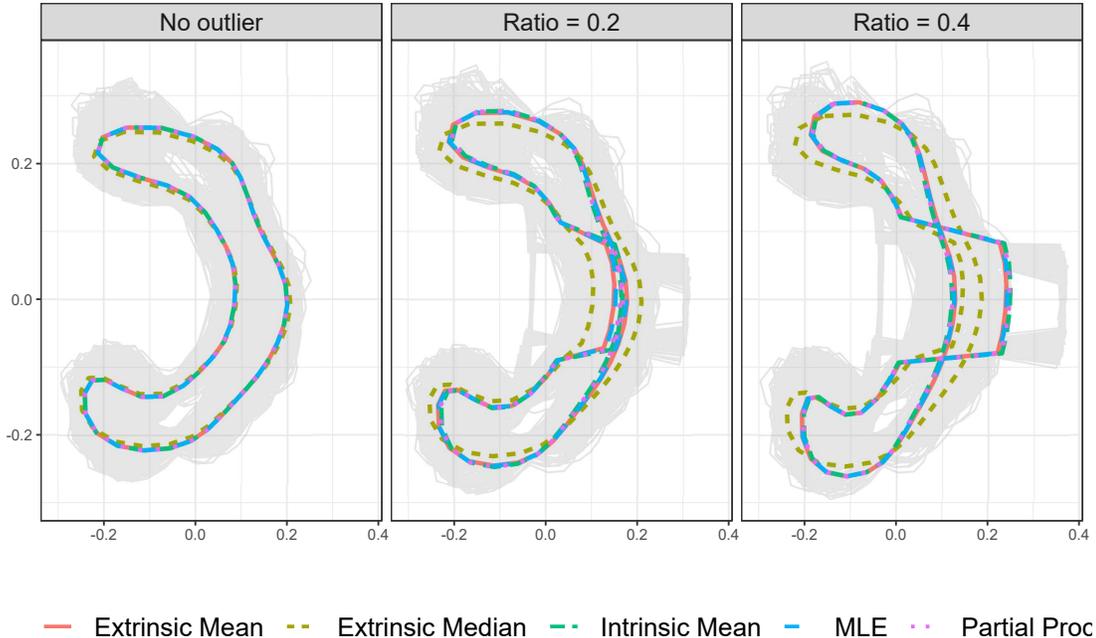}
    \caption{Examples of normal C.C. data and estimated shapes. Each individual shape is displayed in light grey solid line, and the results of different methods are represented by different colors and line styles.}    \label{fig:CC_shape}
\end{figure}
\Cref{fig:CC_shape} shows the CC shapes obtained from several simulated data with different noise level $r=\{0,0.2,0.4\}$. As shown in the left panel ($r=0$), no remarkable difference was observed in estimated shapes between methods. On the other hand, however, the middle and the right panels present that with the exception of the extrinsic median, other methods appeared to be affected by outliers and led to the distortion in the estimated shapes. Importantly, although the deformation of the estimated shape is occurred as well in the extrinsic median at the highest noise level tested, we have seen that it stays much closer to its initial result, than those of the other methods compared.

We now introduce the measure that quantifies the robustness of estimators on the planar shape space. To do this, we let $\overline{\bx}$, $\widehat{\bx}^\ast$ denote the estimated shape obtained from the uncontaminated and contaminated data, respectively. Then the full Procrustes distance between $\overline{\bx}$ and $\widehat{\bx}^\ast$, i.e., $\rho_{FP}(\widehat{\bx}^\ast, \overline{\bx}) = \sqrt{1 - \vert\langle \widehat{\bx}^\ast, \overline{\bx} \rangle\vert^2}$, is considered to assess whether methods can provide the robust estimation without being influenced by outlier values. This appears analogous to that used by the breakdown point in which the Euclidean version of the foregoing quantity, $\Vert T_n(\mathcal{X}^\ast_m) - T_n(\mathcal{X}) \Vert$ is exploited. However, unlike in the case of the breakdown point, which gives the highest fraction of gross outliers in the data that can be handled by an estimator, the smaller value of $\rho_{FP}(\widehat{\bx}^\ast, \overline{\bx})$ implies the method has more resistance to outliers.
Results of this simulation, averaged over 20 replications for the different contamination levels, are presented in \Cref{fig:shape_simulation}. This illustrates that the proposed extrinsic median has a remarkable ability to resist against outliers in the case where the data are contaminated with significant levels of noise. All the methods, however, show deterioration in performance which mainly caused by the squared distance term employed in models. 
\begin{figure}
\centering
    \includegraphics[width=0.9\linewidth]{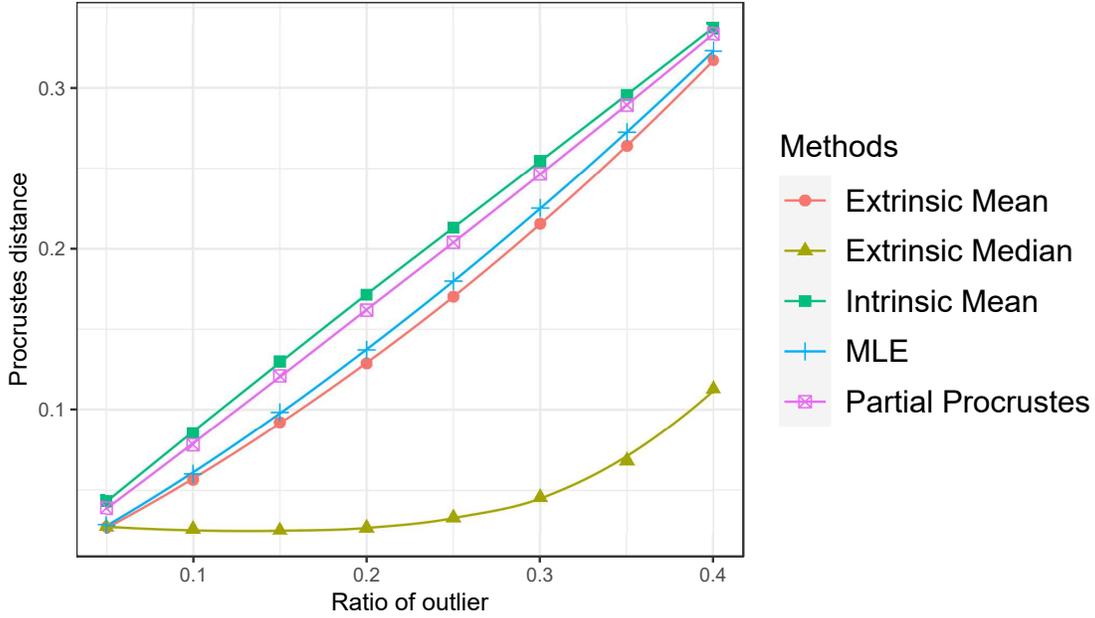}
    \caption{Graphical result of the simulation study. Line plots give the full Procrustes distances, $\rho_{FP}(\widehat{\bx}^\ast, \overline{\bx})$ for different methods as a function of the contamination level $r$.}
    \label{fig:shape_simulation}
\end{figure}



\section{Robust Extrinsic Local Regression} \label{sec:Robust_Reg}

In this section, we present the robust extrinsic local regression. To do this, we first consider a nonparametric regression model $\bY = f_0(\bX) + \bveps$ with a response $\bY$ taking value in $\mathcal{M}$, a Euclidean predictor $\bX \in \mathbb{R}^p$, and $f_0$ an unknown regression function of interest. Suppose we observe $\mathcal{D} = \{(\bx_1,\by_1), \cdots, (\bx_n,\by_n) \}$ consisting of independent and identically distributed copies of $(\bX,\bY)$.
One of the major challenges involved in developing a regression model having a manifold valued response lies in the lack of vector space structure of $\mathcal{M}$, which causes the traditional Euclidean approaches including the least square method not to be obviously applicable. For example, since linear operations are limited on $(\mathcal{M},\rho)$, evaluating the difference between the estimated value and the observed value, i.e., $\by_i - \widehat{f}(\bx_i)$, is not practical. Moreover, the geometrical feasibility of the estimation, i.e., $\widehat{f}(\bx_i) \in \mathcal{M}$, can not be guaranteed unless additional restrictions are imposed on the typical regression models. For the reasons outlined above, there has been a great demand for the development of a regression model having a manifold valued response, and a large body of literature addressing this problem has accumulated over the past two decades \citep{Sh:2009,YuZhLiMa:2012,cozhki:2017}. In particular, the extrinsic local regression (ELR) method has been initially established by \citet{LiBrThHoDu:2017}. More recently, \citet{PeMu:2019} proposed the Fr\'echet regression on general metric spaces by considering the following conditional Fr\'echet mean
\begin{align*}
    F(\px) = \argmin_{\bq \in \mathcal{M}} \int_\mathcal{M} \rho^2 (\bq,\by) \mathcal{Q}(d\by \vert \px),
\end{align*}
where $\mathcal{Q}(\by \vert \px)$ denotes the conditional distribution of $\bY$ given $\bX=\px$. Applications of the above framework is very broad as its usage is not limited to manifolds. However, despite promising progress in developing regression models for a non-Euclidean valued response, all the aforementioned methods commonly suffer from lack of robustness caused by the squared distances. To remedy this problem, we propose the robust extrinsic local regression (RELR), which can be accomplished easily by linking the extrinsic median to a classical nonparametric local kernel regression. The remainder of this section is dedicated to presenting the details of RELR, together with the proposed numerical algorithm.

We begin by introducing the following population robust extrinsic regression function, which extends the notion of the conditional median to manifolds,
\begin{align}
    F_{RE}(\px) &= \argmin_{\bq \in \mathcal{M}} \int_\mathcal{M} \Vert J(\bq) - J(\by) \Vert \mathcal{Q}(d\by \vert \px)\\
    &= \argmin_{\bq \in \mathcal{M}} \int_\mathcal{\widetilde{\mathcal{M}}} \Vert J(\bq) - \bz \Vert \widetilde{\mathcal{Q}}(d\bz \vert \px) \notag,
\end{align}
where $\widetilde{\mathcal{Q}}(\cdot \vert \px) = \mathcal{Q}(\cdot \vert \px) \circ J^{-1}$ is the induced conditional probability measure of $\bY$ given $\bX = \px$ defined on $J(\mathcal{M})$. While the proposed extrinsic approach is similar in spirit to those developed in \citet{LiBrThHoDu:2017}, our work differs in that it makes use of the unsquared extrinsic distance rather than the squared one. The unknown regression function $F(\cdot)$ can be estimated at the evaluation point $\px$ by the classical local polynomial fitting \citep{FaGi:1996}
\begin{align}
    \widehat{F}_{RE}(\px) = J^{-1} \left(\mathcal{P} \bigg(\argmin_{\py \in \mathbb{R}^D} \sum_{i=1}^n \dfrac{K_\bH(\bx_i - \px) \Vert \py-J(\by_i)\Vert}{\sum_{j=1}^n K_\bH(\bx_j - \px)}\bigg)\right).
    \label{eq:Extrinsic_Robust_Reg}
\end{align}
In the above notation, $K_\bH : \mathbb{R}^p \rightarrow \mathbb{R}$ denotes the multivariate kernel function which is defined as $K_\bH(\bu) = \frac{1}{\det{(\bH)}} K(\bH^{-1}\bu)$, where $\bu=(u_1,\cdots,u_p)^\top \in \mathbb{R}^P$, $\bH$ is a $p\times p$ symmetric and positive definite smoothing matrix, and $K: \mathbb{R}^p \rightarrow \mathbb{R}$ satisfies $\int_{\mathbb{R}^p} K(\bu)d\bu = 1, \int_{\mathbb{R}^p} \bu K(\bu)d\bu = 0$, and $\int_{\mathbb{R}^p} \bu^2K(\bu)d\bu < \infty$. Note that the case $\bH =  \operatorname{Diag}(h_1,\cdots,h_p)$ corresponds to using a product kernel obtained by multiplying $p$ univariate kernels with different bandwidths, i.e., $K_\bH(\bu) = \prod_{i=1}^p \frac{1}{h_i} {\mathbf{k}}_i\left(u_i/h_i\right)$. Regarding solving the inner optimization problem in \eqref{eq:Extrinsic_Robust_Reg}, note that since it takes the form of the weighted Fermat-Weber problem, where the weight imposed on the $i$th observation is formulated in terms of the kernel function  $w_i = K_\bH(\bx_i - \px) / \sum_{j=1}^n K_\bH(\bx_j - \px)$, the robust extrinsic local regression can be readily solved by the generalized Weiszfeld's algorithm. 

We now describe the numerical algorithm for obtaining the solution of the localized robust regression estimator. As it has been assumed in the development of the extrinsic median, the non-colinearity of the embedded responses $J(\by_1), \cdots, J(\by_n)$ is required in order to ensure the convergence of the algorithm. We also let $f(\py) = \sum_{i=1}^n w_i \Vert \py-J(\by_i) \Vert$ be the objective function, then by the strict convexity of $f$, the optimal solution is attained at the stationary point $\nabla f(\py)=\sum_{i=1}^n w_i \frac{\py - J(\by_i)}{\Vert \py - J(\by_i)\Vert}\equiv0$. Then, since the optimal $\py^\ast$ satisfies the following equation $\left( \sum_{i=1}^n w_i/\Vert \py^\ast- J(\by_i) \Vert \right) \py^\ast = \sum_{i=1}^n w_i J(\by_i)/ \Vert \py^\ast-J(\by_i)\Vert$, the iterative algorithm for updating $\py$ on the embedded space has the following form
\begin{align}
    \py^{t+1} = \left(\sum_{i=1}^n \frac{w_i}{\Vert \py^t- J(\by_i) \Vert}\right)^{-1} \sum_{i=1}^n \frac{w_i J(\by_i)}{\Vert \py^t-J(\by_i)\Vert}\label{eq:RELR_Weisz}.
\end{align}
In \Cref{al:RELR}, after a simple algebraic calculation, we can reformulate the update rule in \eqref{eq:RELR_Weisz} into the form of the gradient descent with the step size $s^t = \left(\sum_{i=1}^n w_i/\Vert \py^t- J(\by_i) \Vert\right)^{-1}$.
Once finding the optimal solution $\py^\ast$ of the inner minimization problem on the ambient space in $\mathbb{R}^D$, the final regression estimator of $\widehat{F}_{RE}(\px)$ on $\mathcal{M}$ can be straightforwardly obtained by evaluating the projection map of $\py^\ast$ onto the image of the $J(\mathcal{M})$ and taking the inverse map of the embedding. Note that in order to prevent the algorithm from getting stuck on the non optimal embedded points, $\{\py^t\}_{t\geq 0} \not\in \{J(\by_1),\cdots, J(\by_n)\}$ needs to be assumed. 
\begin{algorithm}
\caption{Robust extrinsic local regression (RELR)}
        \begin{algorithmic}[1]
            \Input{$n$ observations $\mathcal{D} = \{(\bx_1,\by_1), \cdots , (\bx_n,\by_n)\}$, evaluation point $\px$}
            \Initialize{$t=0, \py^0 \ \text{and} \ \varepsilon$}
            \While{$\Vert \py^{t+1}-\py^{t}\Vert > \varepsilon$}  
                \State Compute the gradient direction $\nabla f(\py^{t})$
                    \begin{align*}
                        \sum_{i=1}^n \frac{K_\bH(\bx_i-\px)}{\sum_{j=1}^n K_\bH(\bx_j-\px)}\frac{\py^{t}-J(\by_i)}{\Vert \py^{t}-J(\by_i)\Vert}
                    \end{align*}
                \State Compute the step size
                    \begin{align*}
                       s^{t} =  \left(\sum_{i=1}^n \frac{K_\bH(\bx_i-\px)}{\sum_{j=1}^n K_\bH(\bx_j-\px)} \middle/ \Vert \py^{t} - J(\by_i) \Vert\right)^{-1}
                    \end{align*}
                \State Update $\py^{t+1}$
                     \begin{align*}
                        \py^{t+1} = \py^{t} - s^t \cdot \nabla f(\py^{t})
                    \end{align*}
                \State{$t \leftarrow t+1$}
            \EndWhile \State \textbf{end while}
            \Output{Estimated robust estimator $\widehat{F}_{RE}(\px) = J^{-1}(\mathcal{P}(\py^\ast))$} \Comment{$\py^\ast$ the optimal value}
        \end{algorithmic}\label{al:RELR}
\end{algorithm}

We are now in a position to present the convergence analysis of the proposed algorithm. Before proceeding, the following results related to convergence of the classical Weiszfeld's algorithm should be noted.
For nonsmooth convex optimization problems, most existing gradient descent type algorithms are known to converge to the optimal solution at a rate of $\mathcal{O}(1/\sqrt{t})$, however, with the initial value proposed by \citet{VaZh:2001}, one can show that the Weiszfeld's algorithm attains a sublinear convergence rate $\mathcal{O}(1/t)$ \citep[see][]{BeSa:2015}. More surprisingly, using the smooth approximation of the object function enables the algorithm to achieve an accelerated convergence rate $\mathcal{O}(1/t^2)$. Additionally, since the proposed algorithm for solving RELR is performed on Euclidean space, similar algorithmic techniques and convergence properties previously established for the classical Weiszfelds's algorithm can be immediately utilized. 

First, we will henceforth make use of the following initial point of RELR algorithm, which is the extension of the scheme in \citet{VaZh:2001} to a nonparametric regression setting. For $p \in \argmin_{i\in \{1,\cdots,n\}} f(J(\by_i))$, we set the initial value for \Cref{al:RELR} by $\py^0 = \by_p+ t_p \bold{d}_p$, where
\begin{align*}
\begin{dcases}
    R_p &:= \sum_{i \neq p} \frac{K_\bH(\bx_i-\px)}{\sum_{j=1}^n K_\bH(\bx_j-\px)} \frac{J(\by_p)-J(\by_i)}{\Vert J(\by_i)-J(\by_p) \Vert}\\
    \bold{d}_p &= -\frac{R_p}{\Vert R_p \Vert}\\
    t_p &= \frac{\Vert R_p \Vert - K_\bH(\bx_p-\px)/\sum_{j=1}^n K_\bH(\bx_j-\px)}{L(J(\by_p))}.
\end{dcases}
\end{align*}
And the operator $L$ is given by
\begin{align*}
    L(\py) = 
    \begin{dcases}
        \sum_{i=1}^n \frac{K_\bH(\bx_i-\px)/\sum_{j=1}^n K_\bH(\bx_j-\px)}{\Vert \py - J(\by_i) \Vert}, & \py \not\in \{J(\by_1),\cdots, J(\by_n)\}\\
        \sum_{i\neq p} \frac{K_\bH(\bx_i-\px)/\sum_{j=1}^n K_\bH(\bx_j-\px)}{\Vert J(\by_p) - J(\by_i) \Vert}, & \py = J(\by_p) \ (1 \leq p \leq n) \ .
    \end{dcases}
\end{align*}
Using the above initial point, we obtain the following result. \Cref{prop:subliear} is of interest in its own right, since without any further manipulation of the algorithm, the carefully chosen starting value enables \Cref{al:RELR} to achieve the sublinear convergence rate, $\mathcal{O}(1/t)$.
\begin{prop}\label{prop:subliear}
Suppose that all embedded response values are not colinear. Then, for any $t \geq 0$, we have
\begin{align}
    f(\py^t) - f^\ast \leq \frac{L(J(\by_p)) \Vert \py^0 - \py^\ast \Vert^2 }{t \left(\Vert R_p \Vert -\frac{K_\bH(\bx_p-\px)}{\sum_{j=1}^n K_\bH(\bx_j-\px)} \right)^2} \ , \label{eq:sublinear}
\end{align}
where $f^\ast$ is the minimum of $f$.
\end{prop}

\begin{proof}
To prove the sublinear convergence rate of the proposed \Cref{al:RELR}, we have used a collection of results in \citet{BeSa:2015}.
We first need to derive the upper bound of the sequence $\{L(\py^t)\}_{t\geq 0}$, where $\{\py^t\}_{t\geq 0}$ is the sequence generated by the algorithm. For any $i=1,\cdots,n$, and $\py$ satisfying $f(\py) \leq f(\py^0)$, the following inequality $\Vert \py - J(\by_i)\Vert \geq f(J(\by_i)) - f(\py^0)$ is satisfied \citep[Lemma 8.1 in][]{BeSa:2015}. By combining this together with the monotonicity of $f(\py^t) \leq f(\py^0)$ \citep[Corollary 3.1 in][]{BeSa:2015} and $f(J(\by_p)) \leq f(J(\by_i))$, we have $\Vert \py^t - J(\by_i) \Vert \geq f(J(\by_p))-f(\py^0)$. Then by making use of the definition of the operator $L$, we obtain the following result
\begin{align}
    L(\py^t) = \sum_{i=1}^n \frac{\frac{K_\bH(\bx_i-\px)}{\sum_{j=1}^n K_\bH (\bx_j - \px)}}{\Vert \py^t-J(\by_i) \Vert} \leq
    \frac{1}{f(J(\by_p))-f(\py^0)} \leq
    \frac{2L(J(\by_p))}{\left(\Vert R_p \Vert - \frac{K_\bH(\bx_p-\px)}{\sum_{j=1}^n K_\bH (\bx_j - \px)}\right)^2}. \label{eq:L_operator}
\end{align}
In the last inequality, we have use the Lemma 7.1 in \citet{BeSa:2015} that for some $j \in \{1,\cdots,n\}$, $f(J(\by_j)) - f(J(\by_j) + t_j \bold{d}_j) \geq (\Vert R_j \Vert - K_\bH(\bx_j-\px)/\sum_{i=1}^n K_\bH (\bx_i - \px) \Vert)^2/2L(J(\by_j))$. Finally from Lemma 5.2 in \citet{BeSa:2015}, which states $f(\py^{n+1}) - f^\ast \leq L(\py^n)\left( \Vert \py^n - \py^\ast \Vert^2 - \Vert \py^{n+1} - \py^\ast \Vert^2 \right)/2$ and the Fej\'er monotonicity of the sequence generated from Weiszfeld's algorithm (i.e., $\Vert \py^{t+1} - \py \Vert \leq \Vert \py^t - \py \Vert$), the upper bound $f(\py^t)-f^\ast$ can be derived in the following manner.
\begin{align}
    \sum_{n=0}^{t-1} \left(f(\py^{n+1}) - f^\ast \right) & \leq \sum_{n=0}^{t-1} \frac{L(\py^n)}{2} \left( \Vert \py^n - \py^\ast \Vert^2 - \Vert \py^{n+1} - \py^\ast \Vert^2 \right)\notag\\
    & \leq \frac{L(J(\by_p))}{\left( \Vert R_p \Vert - \frac{K_\bH (\bx_p - \px)}{\sum_{j=1}^n K_\bH (\bx_j - \px)} \right)^2} \left( \Vert \py^0 - \py^\ast \Vert^2 - \Vert \py^t - \py^\ast \Vert^2 \right)\notag\\
    & \leq \frac{L(J(\by_p)) \Vert \py^0 - \py^\ast \Vert^2}{\left( \Vert R_p \Vert - \frac{K_\bH (\bx_p - \px)}{\sum_{j=1}^n K_\bH (\bx_j - \px)} \right)^2}.  \label{eq:sublinear_pf}
\end{align}
Since the sequence $\{ f(\py^t)\}_{t\geq0}$ is non increasing, $t (f(\py^t) - f^\ast) \leq \sum_{n=0}^{t-1} \left( f(\py^{n+1}) - f^\ast \right)$ completes the proof.
\end{proof}

Even though we have achieved the improved rate of convergence $\mathcal{O}(1/t)$, there still exists a gap in the order of the convergence rate, as compared to $\mathcal{O}(1/t^2)$ which is commonly attained by the accelerated gradient based methods for solving smooth convex optimization problems.
To further enhance the convergence rate, we here introduce the modified version of algorithm for RELR. Since the slow convergence rate is essentially caused by the inherent nonsmoothness of the objective function $f$, this can be resolved by using the smooth alternative of $f$, that always gives the optimal solution exactly the same as it is supposed to be.
Besides the improved convergence rate, what makes this approach is even more surprising is that for $t \geq 1$ the assumption $(\{\py^t\} \not\in \{J(\by_1),\cdots, J(\by_n)\})$ made on the the sequence generated by \Cref{al:RELR} is not required. Given the lack of knowledge on conditions under which the above assumption can be guaranteed, the smooth approximation method has the practical advantage. 

We finish this section by describing the derivation of the modified Weiszfeld's algorithm for RELR along with the convergence analysis. To do this, we begin by considering the following smooth function $\widetilde{f}_s(\py) : \mathbb{R}^D \rightarrow \mathbb{R}$,  
\begin{align}
\widetilde{f}_s(\py)=\sum_{i=1}^n \frac{K_\bH(\bx_i-\px)}{\sum_{j=1}^n K_\bH(\bx_j-\px)} g_{b_i}(\py-J(\by_i)),\label{eq:smooth_approx}
\end{align}
where
\begin{align}
    g_{b_i}(\py - J(\by_i)) = 
    \begin{dcases}
        \Vert \py - J(\by_i) \Vert, & \Vert \py - J(\by_i) \Vert \geq {b_i}\\
        \frac{\Vert \py - J(\by_i) \Vert^2}{2{b_i}} + \frac{{b_i}}{2}, & \Vert \py - J(\by_i) \Vert < {b_i}\ ,
    \end{dcases}\label{eq:gbi}
\end{align}
and $b_i = f(J(\by_i))- f(\py^0)$. Note first that the function $\widetilde{f}_s(\py)$ is convex and continuously differentiable over $\mathbb{R}^D$ whose gradient is Lipschitz continuous, i.e., $\Vert \nabla\widetilde{f}_s(\py)-\nabla\widetilde{f}_s(\pz)\Vert\leq L_s\Vert\py-\pz\Vert \ \forall \py,\pz \in \mathbb{R}^D$, with the Lipschitz constant 
\begin{align*}
L_s = \sum_{i=1}^n \frac{K_\bH(\bx_i - \px)/\sum_{j=1}^n K_\bH(\bx_j - \px)}{f(J(\by_i))-f(\py^0)} \ .
\end{align*}
Also observing that $g_{b_i}(\py - J(\by_i)) \geq \Vert \py - J(\by_i) \Vert$ \ $\forall \py \in \mathbb{R}^D$, it follows that we have $\widetilde{f}_s(\py) \geq f(\py)$ which indicates $\widetilde{f}_s(\py)$ plays a role of the upper bound of $f(\py)$. Moreover following from Lemma 8.1 in \citet{BeSa:2015}, $\Vert \py^\ast - J(\by_i)\Vert \geq f(J(\by_i) - f(\py^0) = b_i$ holds for $i = 1,\cdots,n$, where $\py^\ast$ is the strict global minimizer of the original objective function $f$. Then, according to the construction in \eqref{eq:gbi}, $g_{b_i}(\py^\ast - J(\by_i))=\Vert \by^\ast - J(\by_i) \Vert$ is always satisfied, which indicates $\widetilde{f}_s(\py^\ast) = f(\py^\ast) < f(\py) \leq \widetilde{f}_s(\py)$. Thus, it is clear to see that the minimizer of the inner optimization problem in \eqref{eq:Extrinsic_Robust_Reg} must also be a global minimizer of \eqref{eq:smooth_approx}, i.e., $\py^\ast=\argmin_{\py}\widetilde{f}_s(\py) = \argmin_{\py}f(\py)$. In the above sense, $\widetilde{f}_s(\py)$ allows us to smoothly approximate the original objective function $f(\py)$. Therefore, rather than directly working with $f(\py)$, we should aim to minimize $\widetilde{f}_s(\py)$, which leads to \Cref{al:fast_RELR}. 
\begin{algorithm}
\caption{Fast Weiszfeld algorithm for RELR}
        \begin{algorithmic}[1]
            \Input{$n$ observations $(X,Y) = \{(\bx_1,\by_1), \cdots , (\bx_n,\by_n)\}$, evaluation point $\px$}
            \Initialize{$s_1=1, \pu^1 = \py^0 \in \mathbb{R}^d \ \text{and} \ \varepsilon$}
            \While{$\Vert \py^{t}-\py^{t-1}\Vert > \varepsilon$} \\
            For $t = 1, 2, \cdots $
                \State Update $\displaystyle \py^t = \pu^t - \frac{1}{L_s}\nabla \widetilde{f}_s(\pu^t)$, where 

                    \begin{align*}
                        \nabla \widetilde{f}_s(\pu^t) = 
                            \sum_{i=1}^n \frac{K_\bH(\bx_i - \px)}{\sum_{j=1}^n K_\bH(\bx_j - \px)} \times
                        \begin{cases}
                            \dfrac{\pu^t - J(\by_i)}{\Vert \pu^t - J(\by_i) \Vert}, \ &\text{if} \ \Vert \pu^t - J(\by_i) \Vert \geq b_i \\\\ 
                            \dfrac{\pu^t - J(\by_i)}{b_i} , \ &\text{if} \ \Vert \pu^t - J(\by_i) \Vert < b_i
                        \end{cases}
                    \end{align*}                    
        \State Update $\displaystyle s_{t+1} = \frac{1 + \sqrt{1 + 4s_t^2}}{2}$ \vspace{.4cm}
        \State Update $\displaystyle \pu^{t+1} = \py^t + \left(\frac{s^t-1}{s^{t+1}}\right) (\py^t - \py^{t-1}) $

            \EndWhile \State \textbf{end while}
            \Output{Estimated robust estimator $\widehat{F}_{RE}(\px) = J^{-1}(\mathcal{P}(\py^\ast))$} \Comment{$y^\ast$ optimal value}
        \end{algorithmic}\label{al:fast_RELR}
\end{algorithm}

Now we let $\{\py^t\}_{t\geq0}$ be the sequence generated by \Cref{al:fast_RELR}, then we obtain
\begin{align}
    \widetilde{f}_s(\py^t) - f^\ast \leq \frac{2  L_s \Vert \py^0 - \py^\ast \Vert^2}{(t+1)^2}.
    \label{eq:fast_Weiszfeld}
\end{align}
The above convergence result follows immediately from Theorem 9.1 in \citet{BeSa:2015}, and see \citet{BeTe:2009} for a proof. Furthermore, in our RELR setting, $L_s$ is bounded from above by
\begin{align*}
    L_s &= \sum_{i=1}^n \frac{K_\bH(\bx_i - \px)/\sum_{j=1}^n K_\bH(\bx_j - \px)}{f(J(\by_i))-f(\py^0)} 
    \leq \sum_{i=1}^n \frac{K_\bH(\bx_i - \px)/\sum_{j=1}^n K_\bH(\bx_j - \px)}{f(J(\by_p))-f(\py^0)} \\ 
    &=  \frac{1}{f(J(\by_p))-f(\py^0)} \leq \frac{2L(J(\by_p))}{ \left( \Vert R_p \Vert - \frac{K_\bH(\bx_p-\px)}{\sum_{j=1}^n K_\bH(\bx_j-\px)} \right)^2} \ ,
\end{align*}
where the last inequality uses \eqref{eq:L_operator}.
Putting all pieces together and using $\widetilde{f}_s(\py)-f^\ast \geq f(\py)-f^\ast$, we see that the fast Weiszfeld algorithm for RELR attains the following convergence rate of $\mathcal{O}(1/t^2)$ :
\begin{align*}
    f(\py^t) - f^\ast \leq  \frac{4 \Vert \py^0 - \py^\ast \Vert^2L(J(\by_p))  }{\left\{(t+1) \left( \Vert R_p \Vert - \frac{K_\bH(\bx_p-\px)}{\sum_{j=1}^n K_\bH(\bx_j-\px)} \right) \right\}^2} \ ,
\end{align*}
which is a substantial improvement on the convergence rate in \eqref{eq:sublinear}.


\section{Application of RELR to the Planar Shape Space}\label{sec:simulation}
In this section, the benefit of the robust extrinsic local regression over the extrinsic regression is demonstrated on the basis of simulation studies. While simulation conducted in this paper has focused only on the case of the planar shape, the proposed method can be applied to other manifolds in a straight forward manner. Recalling from \Cref{sec:app_planarshape}, we let $\bY = (z_1, \cdots ,z_k)$ be the response variable which is a planar shape of $k$-ads defined on $\Sigma_2^k$, and $\bX \in \mathbb{R}^p$ is an Euclidean predictor. By slightly modifying the polar coordinate based scheme in \citet{LiBrThHoDu:2017}, we generate synthetic planar shape data in the following manner : 
\begin{allowdisplaybreaks}
\begin{flalign*}
    & \textbf{Generate Covariate :} \ \bX = (X_1, \cdots, X_p) \ , \ \text{where} \ X_i  \sim \text{Uniform}(a,b) \\
    & \textbf{Coefficient :} \  \bbeta = (\beta_1,\cdots\beta_k) = (1/k^2,\cdots,k/k^2) \in \mathbb{R}^k\\
    & \textbf{Generate Intercept angles :}\ {\boldsymbol{\phi_{0}}}  = ({\boldsymbol{\phi_{0}}}_1,\cdots,{\boldsymbol{\phi_{0}}}_k) = (1/2,\cdots,k/2) \in \mathbb{R}^k \\
    & \textbf{Generate Intercept radius :}\ {\boldsymbol{\gamma_{0}}} = ({\boldsymbol{\gamma_{0}}}_1,\cdots,{\boldsymbol{\gamma_{0}}}_k) = (0.1,\cdots,0.1) \in \mathbb{R}^k\\
    & \textbf{Generate Shape angles :}\ \phi_{j}^\prime \sim \operatorname{Normal}\left({\boldsymbol{\phi_{0}}}_j + \beta_j\sum_{i=1}^p  X_{i}, \sigma_\phi^2\right)\\
    & \textbf{Standardize angles :}\ {\boldsymbol{\phi}} = (\phi_1,\cdots\phi_k), \text{where} \ \phi_j = {\phi_j^\prime}\hspace{-.3cm}\pmod{2\pi}\\
    & \textbf{Generate Shape radius :}\ {\boldsymbol{\gamma}}=(\gamma_1,\cdots,\gamma_k), \ \text{where} \  \gamma_{j} \sim \operatorname{Normal}\left({\boldsymbol{\gamma_{0}}}_j + \beta_j\sum_{i=1}^p  X_{i}, \sigma_\gamma^2\right)\\
    & \textbf{Convert to complex form for the landmark} : 
            z_j = \gamma_j(\cos(\phi_j) + i \sin(\phi_j)).
\end{flalign*}
\end{allowdisplaybreaks}
Further, in order to conduct simulation studies under the outlier contaminated setting, we randomly add fixed number of outliers to the response variable, i.e.,  $\bY^\ast = \bY + {\bf{\Psi}} = (z_1^\ast, \cdots ,z_k^\ast)$. The extreme value of outlier ${\bf\Psi} \in \mathbb{C}^k$ was generated from the $k$-dimensional complex normal distribution, $\mathbb{C}N(\bmu,\bf{\Gamma})$, where $\bmu = E({\bf\Psi})$ and ${\bf\Gamma} = E\left(({\bf\Psi}-\bmu)({\bf\Psi}-\bmu)^H\right)$.
We note that since the contaminated response $\bY^\ast$ is no longer an element in $\Sigma_2^k$, both translation and scale effects have to be filtered out. To help understand the process by which data are generated, the representative illustration of the simulated data is presented on the left panel of \Cref{fig:sim_shape_reg}. For illustrative purpose, planar shape data is generated under the univariate setting, and only the first landmark is contaminated. In the right panel of the same figure, the estimated curves are given, along with true underlying function. As shown in the figure, the curve obtained from the usual ELR (red) substantially deviates from the true curve (blue), while the curve corresponding to the RELR (green) is almost overlapped with the supposed true value.

\begin{figure}[h]
\centering
    \includegraphics[width=0.48\linewidth]{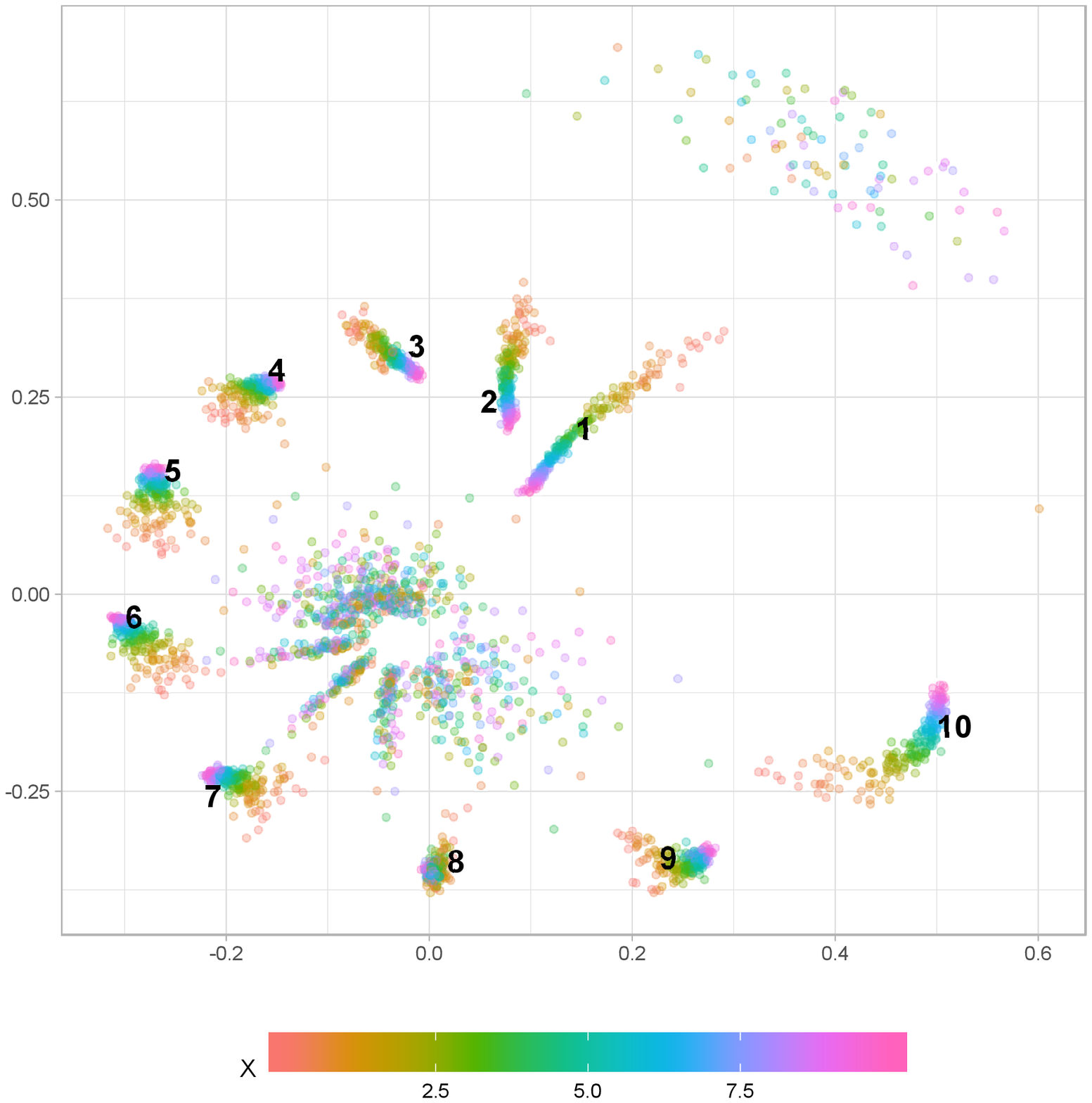}
    \includegraphics[width=0.48\linewidth]{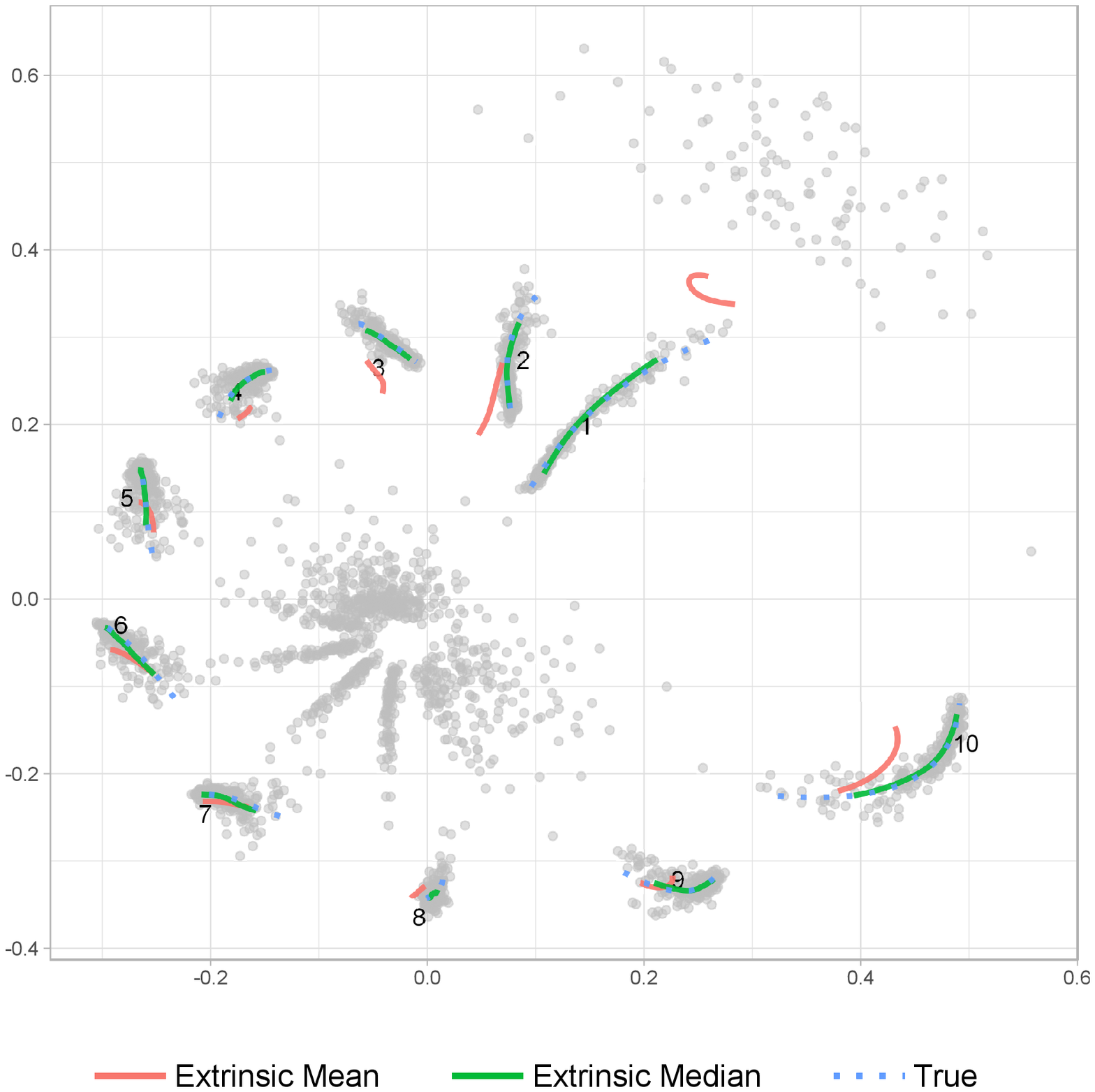}
    \caption{Left : the example of the simulated data (univariate case) with $r = 0.2$, each point is colored according to the value of the predictor $X_1$. Right : The result of estimations. The optimal bandwidths of RELR and ELR, selected by the 5-fold cross validation are $h_{\text{Med}} = 1.37$ and $h_{\text{Mean}} = 2.27$, respectively.}
    \label{fig:sim_shape_reg}
\end{figure}

One important remaining issue of the proposed RELR model that has not been highlighted in the previous section is the bandwidth selection. It is well known that the performance of the local polynomial type method significantly relies on a tuning parameter $h$, called the bandwidth which plays a crucial role in controlling the degree of smoothing. To be specific, large $h$ value leads to a smooth estimation, but by failing to account for a local variation it may introduce a significant estimation bias, whereas small $h$ produces a jagged estimation, leading to a large variance. Thus it should be properly selected to balance the trade-off between variance and squared bias. Throughout the simulation, we consider the smoothing matrix that gives the same bandwidth in all $p$ dimensions, i.e., $\bH = h {\bf{I}}_p$. Though bandwidth selection methods have been extensively studied in the early days of nonparametric regression modeling, in this paper we adopt 5-fold cross validation for the sake of simplicity.

The performance of the RELR is evaluated by comparing the results with those achieved by ELR in terms of two different measures associated with the full Procrustes distance, $\rho_{\text{FP}} = \sqrt{\left(1- \vert \langle \bz_1, \bz_2 \rangle \vert^2 \right)}$, where $\bz_1,\bz_2 \in \Sigma_2^k$. Firstly, we consider $\operatorname{MD}_{\text{obs}} = \sum_{i=1}^n \rho_{\text{FP}}(\by_i , \widehat{f}(\bx_i))/n$, that measures difference between the estimated value $\widehat{f}(\bx_i)$ and the observed value $\by_i$. Moreover, to assess whether the estimator has the benefit of capturing the true signal, it is more appropriate to examine the following root mean squared error like measure $\operatorname{RMSE}_{\text{true}} = \sqrt{\sum_{i=1}^n \rho_{\text{FP}}(f_0(\bx_i) , \widehat{f}(\bx_i))^2/n}$, which quantifies the difference between the predictor and the true value $f_0(\bx_i)$. To investigate how methods are affected by outliers, the contamination level $r$ were varied on an evenly spaced grid over $[0,0.3]$. The values (averaged over 20 replications per each setting) of $\operatorname{MD}_{\text{obs}}$ and $\operatorname{RMSE}_{\text{true}}$ with $n = 200, p=1$, are presented in \Cref{fig:sim_shape_rmse} (left and right panel, respectively). 
\begin{figure}
\centering
    \includegraphics[width=0.9\linewidth]{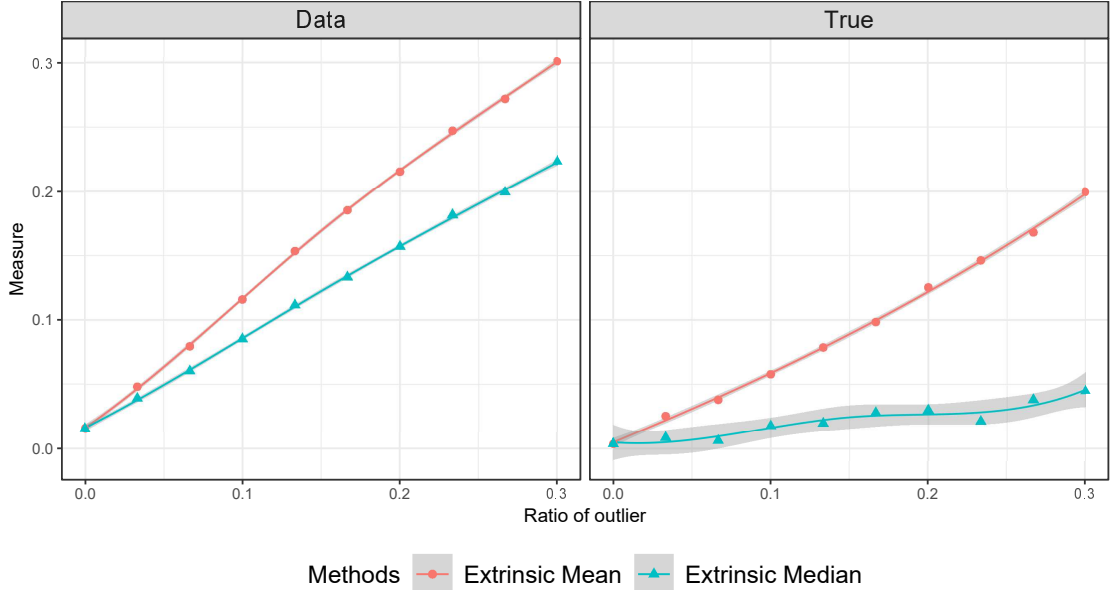}
    \caption{Results of the univariate case, as a function of the contamination level $r = [0,0.3]$. Left : Averaged value of $\operatorname{MD}_{\text{obs}}$. 
    Right: Averaged value of RMSE against underlying true regression function $f_0(x_i)$. To assist in the visualization, we also give a loess smooth with Monte-Carlo
standard error.}
    \label{fig:sim_shape_rmse}
\end{figure}

Overall, the performance curves, obtained across the experimental conditions, are roughly linear and slope upward from left to right as contamination rate runs from 0 to 0.3, which indicates performances of both RELR and ELR degrade as $r$ increases. We also have found that, values of both $\operatorname{MD}_{\text{obs}}$ and $\operatorname{RMSE}_{\text{true}}$ corresponding to RELR are consistently lower than those of ELR, except for $r=0$. The result from the left panel makes intuitive sense, as RELR minimizes the empirical risk function associated with the unsquared extrinsic distance which only differs from the full Procrustes distance by a multiplicative constant. In the right panel, however, the value of the slope obtained from ELR, inclined toward the upper right point $(0.3, 0.2)$ starting from the origin, is 1.67 which is approximately four times greater than that of RELR. This clearly suggests that RELR outperforms ELR by a wide margin in identifying true patterns of shape changes that are masked by noise. As shown in \Cref{tab:sim_shape_reg}, similar results were obtained from further experiments in a multivariate setting ($p=3$) with different samples sizes. In this simulation study, the proposed RELR is seen to universally perform very well across all examined simulation scenarios. However, contrary to RELR, it appears that the performance of ELR is prone to be adversely affected by the presence of outliers and noises, which is consistent and in line with what we would have been expected from the previous location estimation performed in \Cref{sec:app_planarshape}.

\begin{table}
    \centering
    \begin{tabular}{lllllllllllllllll}
    \toprule
         & & & \multicolumn{4}{c}{Ratio of outliler}\\
        \cmidrule(lr){4-8}
         $N$ & Measures & Methods & 0 & 0.05 & 0.1 & 0.2 & 0.3\\
        \midrule
        \multirow{4}{*}{50}
        &\multirow{2}{*}{$\operatorname{MD}_{\text{obs}}$}
        & ELR &  0.0323 & 0.0716 & 0.1389 & 0.2481 & 0.3505 \\ 
     & & RELR  & 0.0320 & 0.0662 & 0.1192 & 0.2052 & 0.2806\\
        \cmidrule(lr){4-8}
         &\multirow{2}{*}{$\operatorname{RMSE}_{\text{true}}$} 
        & ELR & 0.0233 & 0.0348 & 0.0692 & 0.1446 & 0.2326 \\ 
         & & RELR  & 0.0242 & 0.0257 & 0.0400 & 0.0669 &  0.0781 \\
        \cmidrule(lr){2-8}
        \multirow{4}{*}{100}
        &\multirow{2}{*}{$\operatorname{MD}_{\text{obs}}$}
        & ELR & 0.0333 & 0.0810 & 0.1306 & 0.2441 & 0.3524 \\ 
         & & RELR & 0.0338 & 0.0758 & 0.1200 & 0.2098 & 0.2884\\
        \cmidrule(lr){4-8}
         &\multirow{2}{*}{$\operatorname{RMSE}_{\text{true}}$} 
        & ELR & 0.0242 & 0.0358 & 0.0550 & 0.1279 & 0.2219 \\ 
         & & RELR  & 0.0255 & 0.0274 & 0.0414 & 0.0617 & 0.0698\\
        \cmidrule(lr){2-8}
         \multirow{4}{*}{200}
        &\multirow{2}{*}{$\operatorname{MD}_{\text{obs}}$}
        & ELR &  0.0336 & 0.0802 & 0.1351 & 0.2486 & 0.3516\\ 
         & & RELR  & 0.0341 & 0.0768 & 0.1235 & 0.2126 & 0.2880 \\
        \cmidrule(lr){4-8}
         &\multirow{2}{*}{$\operatorname{RMSE}_{\text{true}}$} 
        & ELR & 0.0237 & 0.0327 & 0.0566 & 0.1226 & 0.2093 \\ 
         & & RELR  & 0.0248 & 0.0277 & 0.0386 & 0.0586 & 0.0720 \\
        \bottomrule
    \end{tabular}
    \caption{Comparisons of RELR and ELR in terms of $\operatorname{MD}_{\text{obs}}$ and $\operatorname{RMSE}_{\text{true}}$.}
    \label{tab:sim_shape_reg}
\end{table}
\section{Discussion}\label{sec:conclusion}

In this paper, we have proposed the robust statistical methods on manifolds and demonstrated that when outliers exist in the dataset, they are capable of achieving considerable improvements over existing methods. Building upon the idea of geometric median and the extrinsic framework, our method takes the full advantage of what the two approaches can provide. (i) the robustness property is attained by employing the unsquared extrinsic distance induced by the Euclidean embedding, which prevents the estimator from amplifying the effects of noise and outliers. (ii) Our approach also can be universally adapted to any manifold, on which the proper Euclidean embedding is available. For example, the RELR can be straight-forwardly implemented into the space of $p\times p$ symmetric positive definite (SPD) matrices, which especially emerged as the form of data in neuro imaging. To be specific, $3\times3$ SPD matrices have arisen as data elements in diffusion tensor imaging (DTI). In this space, the equivariant embedding is given by the Riemannian logarithm map from the SPD matrices to symmetric matrices, $\log : \text{SPD}(3) \rightarrow \text{Sym}(3);\bX = \bU {\bf\Lambda} \bU^{-1} \in \text{SPD}(3) \mapsto \text{log}(\bX) = \bU \log{({\bf\Lambda})}\bU^{-1}$, and the extrinsic distance is defined as the Frobenius norm, i.e., $\rho_E(\bX_1,\bX_2) = \Vert \log{(\bX_1)} - \log{(\bX_2)} \Vert_F$.

While our main focus was on developing robust statistical methods for estimating the central location and regression problem on manifolds, some important problems still remain to be further investigated as part of future work. Now, we end the paper by outlining some promising future directions for research that may immediately benefit from our proposed framework. First, borrowing ideas from the method of $K$-medians algorithms \citep{CaCeMo:2012}, clustering on manifolds can be carried out in a more robust manner. Second, an intriguing research question raised by our study is the possible extension of the notion of the quantile to manifolds. Unlike a measure of the central tendency (mean and median), the generalization of the quantile on non-Euclidean space is nontrivial, because it has to take into account the direction and magnitude of changes from the central location. To the best of our knowledge, there has been no reported research in this context yet, which may be in part due to the difficulty in defining the direction on the surface of manifolds. The challenge above may be tackled by utilizing the geometric quantile \citep{Ch:1996}, which has the geometric median as a special case and our extrinsic framework. We expect that it will be especially useful in medical imaging analysis. Because using the quantile information enables us to quantify pathological state or the abnormality of a certain organ shape, it will be a useful tool for diagnosing and prognosticating critically ill patients

\bibliographystyle{apa}
\bibliography{Robust_Extrinsic}

\end{document}